\documentclass[reqno]{amsart}

\usepackage{mathrsfs}
\usepackage{latexsym}
\usepackage{hyperref}

\newtheoremstyle{myplain}{}{}{\it}
{0pt}{\scshape}{}{ }{\thmname{#1}\thmnumber{ #2}\thmnote{ (#3)}}
\newtheoremstyle{mydefinition}{}{}{}
{0pt}{\scshape}{}{ }{\thmname{#1}\thmnumber{ #2}\thmnote{ (#3)}}

\theoremstyle{myplain}

	\newtheorem{Def}{Definition}[section]
        \newtheorem{Lem}[Def]{Lemma}
        \newtheorem{theo}[Def]{Theorem}
        \newtheorem{prop}[Def]{Proposition}
        \newtheorem{rem}[Def]{Remark}
        \newtheorem{hyp}[Def]{Hypothesis}
        \newtheorem{cor}[Def]{Corollary}

\newcommand{\skp}[2]{\mbox{$\left\langle #1\, , \, #2\right\rangle$}}
\newcommand{\skpd}[2]{\mbox{$\left\langle #1\, ,\,#2\right\rangle_{\ell^2}$}}
\newcommand{\skpT}[2]{\mbox{$\left\langle #1\, ,\,#2\right\rangle_{\mathbb T}$}}
\newcommand{\skpR}[2]{\mbox{$\left\langle #1\, ,\,#2\right\rangle_{L^2}$}}

\newcommand{\natop}[2]{\genfrac{}{}{0pt}{}{#1}{#2}}

\DeclareMathOperator{\Span}{span}

\DeclareMathOperator{\rank}{rank}

\DeclareMathOperator{\supp}{supp}

\DeclareMathOperator{\Op}{Op}
\DeclareMathOperator{\loc}{loc}

\newcommand{\id}{\mathbf{1}}

\numberwithin{equation}{section}

\newcommand{\beqa}{\begin{eqnarray*}}
\newcommand{\eeqa}{\end{eqnarray*}}
\renewcommand{\hat}{\widehat}
\newcommand{\bauf}{\begin{itemize}}
\newcommand{\eauf}{\end{itemize}}
\newcommand{\be}{\begin{equation}}
\newcommand{\ee}{\end{equation}}
\newcommand{\ben}{\begin{enumerate}}
\newcommand{\een}{\end{enumerate}}
\newcommand{\ra}{\rightarrow}

\newcommand{\ep}{\varepsilon}

\newcommand{\R}{{\mathbb R} }
\newcommand{\Z}{{\mathbb Z}}
\newcommand{\C}{{\mathbb C}}
\newcommand{\N}{{\mathbb N}}
\newcommand{\T}{{\mathbb T}}

\newcommand{\disk}{(\varepsilon {\mathbb Z})^d}

\newcommand{\Ce}{\mathscr C}

\newcommand{\De}{\mathscr D}

\title{Harmonic approximation of difference operators}

\author{Markus Klein \and Elke Rosenberger}

\address{Markus Klein\\ Universit\"at Potsdam\\ Institut f\"ur Mathematik \\ Am Neuen Palais 10\\ 14469 Potsdam }
\email{mklein@math.uni-potsdam.de}
\address{
Elke Rosenberger\\ Universit\"at Potsdam\\ Institut f\"ur Mathematik \\ Am Neuen Palais 10\\ 14469 Potsdam}
\email{erosen@rz.uni-potsdam.de}

\date{\today}


\begin{document}

\begin{abstract}
For a general class of difference operators $H_\ep = T_\ep + V_\ep$ on $\ell^2(\disk)$,
where $V_\ep$ is a multi-well potential and $\ep$ is a small parameter, we
analyze the asymptotic behavior as $\ep\to 0$ of the (low-lying) eigenvalues and
eigenfunctions. We show that the first
$n$ eigenvalues of $H_\ep$ converge to
the first $n$ eigenvalues of the direct sum of harmonic oscillators on $\R^d$ located
at the several wells. Our proof is microlocal.
\end{abstract}

\maketitle

\section{Introduction}

The central topic of this paper is the investigation of a rather general class of families of
difference operators $H_\ep$ on the Hilbert space $\ell^2(\disk)$, as the small parameter  $\ep>0$
tends to zero. The operator $H_\ep$ is given by
\begin{align} \label{Hepein}
H_\ep = (T_\ep + V_\ep ),  \quad&\text{where}\quad
T_\ep  = \sum_{\gamma\in\disk} a_\gamma \tau_\gamma ,\\
(\tau_\gamma u)(x) = u(x+\gamma) \quad &\text{and}\quad (a_\gamma u)(x) := a_\gamma(x,\ep) u(x) \quad \mbox{for} \quad x,\gamma\in\disk
\end{align}
where $V_\ep$ is a multiplication operator, which in leading order is given by
$V_0 \in \Ce^\infty (\R^d)$.

We will show that the
low lying spectrum of $H_\ep$ on $\ell^2(\disk)$ is in the limit $\ep\to 0$ asymptotically given by
the spectrum of an adapted harmonic oscillator
on $\mathscr{L}^2(\R^d)$.
We remark that the limit $\ep \to 0$ is analog to the semiclassical limit $\hbar \to 0$
for the Schr\"odinger operator
$-\hbar^2 \Delta + V$. The central result of this paper (the validity of the harmonic approximation) is the first
basic step in any WKB-theory for the Schr\"odinger operator (see e.g. Simon \cite{Si1}, Helffer-Sj\"ostrand \cite{hesjo}).
In our case, this basic step is considerably more difficult. The discrete kinetic operator $T_\ep$ is not a local
operator (in particular, not a differential operator). Furthermore, $H_\ep$ and the approximating
harmonic oscillator act on different spaces. We remark that this is in fact crucial: Letting $H_\ep$ act on
$L^2(\R^d)$ would lead to infinite multiplicity of the point spectrum. In addition, the proofs for the Schr\"odinger operator in Simon (\cite{Si1},
\cite{simon}) and Helffer-Sj\"ostrand \cite{hesjo} use special identities for differential operators of second order.
To overcome these difficulties, we use a microlocal approach. The basic theorems necessary for our analysis are
proven in Appendix \ref{symbols}.

This paper is based on the thesis Rosenberger \cite{thesis}. It is the second in a series of papers (see Klein-Rosenberger \cite{roklein}); the aim is to develop an analytic approach
to the semiclassical eigenvalue problem and tunneling for $H_\ep$ which is comparable  in detail and
precision
to the well known analysis for the Schr\"odinger operator (see Simon \cite{Si1}, \cite{simon} and
Helffer-Sj\"ostrand \cite{hesjo}). Our motivation comes from
stochastic problems (see Klein-Rosenberger \cite{roklein}, Bovier-Eckhoff-Gayrard-Klein \cite{begk1}, \cite{begk2},
Baake-Baake-Bovier-Klein \cite{baake}). A large class of
discrete Markov chains analyzed in \cite{begk2}
with probabilistic
techniques falls into the framework of difference operators treated in this article.

We assume

\begin{hyp}\label{hypsatz}
\begin{enumerate}
\item The coefficients $a_\gamma(x, \ep)$ in \eqref{Hepein} are
functions
\begin{equation}\label{agammafunk}
a: \disk \times \R^d \times (0,1] \ra \R\, , \qquad (\gamma, x,
\ep) \mapsto a_\gamma(x,\ep)\, ,
\end{equation}
satisfying the following conditions:
\ben
\item[(i)] They have an
expansion
\begin{equation}\label{agammaexp}
a_\gamma(x,\ep) = a_\gamma^{(0)}(x) + \ep \, a_\gamma^{(1)}(x) +
R^{(2)}_\gamma (x, \ep)\, ,
\end{equation}
where $a_\gamma^{(i)}\in\Ce^\infty(\R^d)$ and $|a_\gamma^{(j)}(x)
- a_\gamma^{(j)}(x+ h)| = O(|h|)$ for $j=0,1$ uniformly with respect to $\gamma\in\disk$ and $x\in\R^d$.
Furthermore $R^{(2)}_\gamma \in\Ce^\infty(\R^d\times (0,1])$ for
all $\gamma\in\disk$.
\item[(ii)] $\sum_\gamma a_\gamma^{(0)}  = 0$ and $a_\gamma^{(0)}
\leq 0$ for $\gamma \neq 0$
\item[(iii)] $a_\gamma(x, \ep) =
a_{-\gamma}(x+\gamma, \ep)$ for $x \in \R^d, \gamma \in \disk$
\item[(iv)] For any $n\in \N$ and $\alpha\in\N^d$ there exists a $C>0$ such that for $j=0,1$
uniformly with respect to $x\in\disk$ and $\ep$
\begin{equation}\label{abfallagamma}
\| \, |\tfrac{.}{\ep}|^n\partial^\alpha_x a^{(j)}_.(x)\|_{\ell_\gamma^2(\disk)}\leq C  \qquad\text{and} \qquad
\|\, |\tfrac{.}{\ep}|^n \partial^\alpha_xR^{(2)}_.(x,\ep)\|_{\ell^2_\gamma(\disk)} \leq C\ep^2\; .
\end{equation}
\item[(v)] $\Span \{\gamma\in\disk\,|\, a^{(0)}_\gamma(x) <0\}= \R^d$ for all $x\in\R^d$.
\een
\item
\ben
\item[(i)] The potential energy $V_\ep$ is the restriction to $\disk$ of a
function $\hat{V}_\ep\in\Ce^\infty (\R^d)$, which has an expansion
\begin{equation}\label{Veps}
\hat{V}_{\ep}(x) =  V_0(x) + \ep \,V_1(x) + R_2(x;\ep) \, ,
\end{equation}
where $V_0, V_1\in\Ce^\infty(\R^d)$, $R_{2}\in \Ce^\infty (\R^d\times (0,\ep_0])$
for some $\ep_0>0$ and
for any compact set $K\subset \R^d$
there exists a constant $C_K$ such that $\sup_{x\in K} |R_{2}(x;\ep)|\leq C_K \ep^{2}$.
\item[(ii)] $V_\ep$ is polynomially bounded and there exist constants $R, C > 0$ such that
$V_\ep(x) > C$ for all $|x| \geq R$ and $\ep\in(0,\ep_0]$.
\item[(iii)]
$V_0\geq 0$ and it takes the value 0 only at a finite number of points
$\{x_j\}_{j=1}^m$, where
its Hessian
\begin{equation}\label{Hessische}
({\tilde A}_{\nu\mu}^j) := \frac{1}{2}\left(\frac{\partial^2 V_0}{\partial x_\nu\partial x_\mu}(x_j)\right)
\end{equation}
is positive definite (i.e. the absolute minima are
non-degenerate). We call the minima $\{x_j\}_{j=1}^m$ of $V_0$ potential wells.
\een
\een
\end{hyp}

We set for $\ep\in (0,1]$
\begin{equation}\label{talsexp}
t(x, \xi; \ep) := \sum_{\gamma\in\disk} a_\gamma (x, \ep) \exp \left(-\frac{i}{\ep}\gamma\cdot\xi\right)\, , \quad
x\in\R^d, \xi\in\T^d (=\R^d/(2\pi)\Z^d)\; ,
\end{equation}
and denote the function on $\R^{2d}\times[0,1)$, which is $2\pi$-periodic with respect to $\xi$ by $t$ as well.
The expansion \eqref{agammaexp} of
$a_\gamma(x, \ep)$ leads to the definition
\begin{align}\label{texpand}
t(x, \xi; \ep) &= t_0 (x,\xi) + \ep\,  t_1(x,\xi) + t_2(x, \xi; \ep)\; ,\qquad\text{with}\\
t_j(x,\xi) &:= \sum_{\gamma\in\disk}\, a_\gamma^j(x) e^{-\frac{i}{\ep}\gamma\xi}\, , \quad j=0,1 \nonumber\\
t_2(x,\xi,\ep) &:= \sum_{\gamma\in\disk} R_\gamma^{(2)}(x,\ep) e^{-\frac{i}{\ep}\gamma\xi}\; .\nonumber
\end{align}

\setcounter{Def}{0}

\begin{hyp}
\ben
\item[(c)] At the minima $x_j$ of $V_0$, we assume that $t_0$ defined in \eqref{texpand} fulfills
\[ t_0(x_j, \xi) >0\, , \quad\text{if} \quad |\xi| >0 \, .\]
\een
\end{hyp}

\begin{rem}\label{talskin}
It follows from (the proof of) Klein-Rosenberger \cite{roklein}, Lemma 1.2,
that under the assumptions given in Hypothesis \ref{hypsatz}:
\ben
\item $t\in \Ce^\infty(\R^d\times \T^d\times [0,1))$ and
$\sup_{x,\xi}|\partial_x^\alpha\partial_\xi^\beta
t(x,\xi,\ep)|\leq C_{\alpha,\beta}$ for all $\alpha,\beta\in\N^d$
uniformly with respect to $\ep$. Moreover $t_0$ and $t_1$ are bounded and
$\sup_{x,\xi}|t_2(x,\xi,\ep)| = O(\ep^2)$.
\item At $\xi=0$, for fixed $x\in\R^d$ the function $t_0$ defined in \eqref{texpand} has an expansion
\begin{equation}\label{kinen}
t_0(x,\xi) = \skp{\xi}{B(x)\xi} + O\left(|\xi|^4\right)\qquad\text{as}\;\; |\xi|\to 0\, ,
\end{equation}
where $B:\R^d\rightarrow\mathcal{M}(d\times d,\R)$ is
positive definite and symmetric. By straightforward calculations one gets
\begin{equation}\label{Bnumu}
B_{\nu\mu}(x) = -\frac{1}{2\ep^2} \sum_{\gamma\in\disk} a_\gamma^{(0)}(x) \gamma_\nu\gamma_\mu\; .
\end{equation}
\item By Hypothesis \ref{hypsatz} (a)(iii) and since the $a_\gamma$ are real,  the operator $T_\ep$ defined in \eqref{Hepein} is symmetric.
In the probabilistic context, which is our main motivation, the latter is a standard reversibility condition while the
former ist automatic for a Markov chain.
Moreover, $T_\ep$ is bounded (uniformly in $\ep$) by condition (a)(iv) and
bounded from below by $-C\ep$ for some $C>0$ by condition (a)(iv),(iii) and (ii).
\item A combination of the expansion \eqref{agammaexp} and the reversibility condition (a)(iii) leads to
the fact that the $2\pi$-periodic function $\R^d\ni\xi\mapsto t_0(x,\xi)$ is even.
\item Since $T_\ep$ is bounded, $H_\ep=T_\ep + V_\ep$ defined in
\eqref{Hepein} possesses a self adjoint realization on the maximal domain of $V_\ep$. Abusing notation, we
shall denote this realization also by $H_\ep$ and its domain by $\De(H_\ep)\subset \ell^2\left(\disk\right)$. The associated symbol is denoted by $h(x,\xi;\ep)$. Clearly, $H_\ep$ commutes with complex conjugation.
\een
\end{rem}

We will use the notation
\begin{equation}\label{agammaunep}
\tilde{a}:\Z^d\times\R^d\ni (\eta,x)\mapsto \tilde{a}_\eta(x) := a_{\ep\eta}^{(0)}(x)\in\R
\end{equation}
and set
\begin{equation}\label{tildehnull}
\tilde{h}_0 (x,\xi):= -h_0(x,i\xi) = {\tilde t}_0(x,\xi) - V_0(x)\,:\, \R^{2d} \ra \R\; ,
\end{equation}
where by Remark \ref{talskin} (d)
\begin{equation}\label{tildetdef}
\tilde{t}_0(x,\xi):= -t_0(x,i\xi) = -\sum_{\eta\in\Z^d} \tilde{a}_\eta(x) \cosh \left(\eta\cdot \xi\right) \; .
\end{equation}

The main result of this paper is the following theorem:

\begin{theo}\label{satz1}
Let $H_\ep$ be an operator satisfying Hypothesis \ref{hypsatz} and let
$A^j:= B_{j}^{\frac{1}{2}}{\tilde A}^jB_{j}^{\frac{1}{2}}$, where $\tilde{A}^j$ is given in \eqref{Hessische}
and $B_{j}= B(x_j)$ is defined in \eqref{kinen}.\\
We denote by
\begin{equation}\label{Kj}
K_j := -\Delta + \skp{x}{A^j x} +  V_1(x_j) + t_1(x_j,0)\, , \qquad j=1,\ldots m
\end{equation}
the self adjoint operators on $L^2\left({\mathbb R}^d\right)$ defined by
Friedrich extension
and set $K := \bigoplus_{j=1}^m K_j$ (which is self adjoint on
$\bigoplus_{j=1}^mL^2\left({\mathbb R}^d\right)$).

Then for any fixed $n\in \N^*$ and $\ep$ sufficiently small, $H_{\ep}$ has at least $n$ eigenvalues.
Counting multiplicity, we denote for $k\in\N^*$ the $k$-th eigenvalue of $K$ by $e_k$ and
the $k$-th eigenvalue of $H_{\ep}$ by $E_k(\ep)$ (ordered
by magnitude).
Then, as $\ep\to 0$,
\begin{equation}\label{limen}
 E_k (\ep)= \ep e_k + O\left(\ep^{\frac{6}{5}}\right)\; .
 \end{equation}
\end{theo}

We remark that (under additional assumptions) Theorem \ref{satz1} considerably sharpens Theorem 1 in Baake-Baake-Bovier-Klein \cite{baake}.

The strategy of the proof of Theorem \ref{satz1} is to restrict the Hamilton operator
$H_\ep$ to small $\ep^{\frac{2}{5}}$-scaled neighborhoods of its
critical points in $x$ and $\xi$, i.e. to neighborhoods of
$\left\{(x_j,0)\right\}_{j=1}^m$ in phase space. Then restricted to these regions, the difference operator
can be compared with a corresponding differential operator acting on
$L^2(\R^d)$.

We follow in part the ideas of the proof of Theorem 11.1 in
Cycon-Froese-Kirsch-Simon \cite{simon} on the quasi-classical
eigenvalue limit of a Schr\"odinger operator. But in contrast to
this proof, our difference
operator $T_\ep$ depends on both position and momentum and
acts on a different space than the harmonic oscillator. The first
step of the proof consists in localizing the operator simultaneously with
respect to $x$ and $\xi$, which is done by use of a version of
microlocal calculus adapted to the discrete setting as introduced in
Definition \ref{symbolspaces}. These localized operators still act on $\ell^2(\disk)$. The second step consists in
comparing the localized operators on $\ell^2(\disk)$ with the associated localized operators
on $L^2(\R^d)$, which are standard pseudo-differential operators. With these preparations, the remaining part
of the proof follows closely the arguments in Simon \cite{Si1}.\\

The plan of the paper is as follows.
We introduce in Section \ref{normscalar} some notations, define symbol-spaces on $\R^d\times\T^d$
and associated operators and state some essential results concerning these symbols and operators.
In Section \ref{beweis} we state and prove lemmata, which are essential ingredients for
the proof of Theorem \ref{satz1}. Proposition \ref{micnorm}, Lemma \ref{1zer} and Lemma \ref{chiphi} contain the
main estimates on the error introduced by localizing the relevant operators on $\ell^2(\disk)$ and
$L^2(\R^d)$. Lemma \ref{diskkont} estimates the difference between these operators.
The proof of Theorem \ref{satz1} is finally given in Section \ref{sec223}.
Appendix \ref{symbols} is concerned with pseudo-differential operators in the discrete setting.
In particular, we collect some properties of symbols and
prove the $\ell^2$-continuity of pseudo-differential operators
with symbols in $S_\delta^r(\id)(\R^d\times \T^d)$ (see Definition \ref{symbolspaces}). In
Section \ref{persson} we show an analog of the Theorem of Persson for
some class of difference operators.

\section{Notations and Preliminaries}\label{normscalar}

For $\ep>0$, we consider $\ell^2\left((\ep{\mathbb Z})^d\right)$, the space of square summable functions
on the $\ep$-scaled lattice, with scalar
product
\begin{equation}\label{skpd}
\skpd{u}{v} := \sum_{x\in \disk} {\bar u}(x)v(x) ,\qquad u,v\in \ell^2\left(\disk\right)\: .
\end{equation}
Denoting the $d$-dimensional
torus by $\T^d:= \R^d/(2\pi )\Z^d$, we identify functions in $L^2(\T^d)$ with
periodic functions in $L^2_{\loc}(\R^d)$. Then
\begin{equation} \label{skpT}
\skpT{f}{g} :=
\int_{[-\pi,\pi]^d} {\bar f}(\xi)g(\xi)\, d\xi,
\qquad
\end{equation}
denotes the scalar product in $L^2(\T^d)$.
We denote the associated norms by $ \|\,.\,\|_{\ell^2}$ and $\|\,.\,\|_{\T}$.

The discrete Fourier transform ${\mathscr F}_\ep:L^2\left(\T^d\right)
\to \ell^2\left((\ep{\mathbb Z})^d\right)$ is defined by
\begin{equation}\label{Fou}
({\mathscr F}_\ep f)(x) := \frac{1}{\sqrt{2\pi}^d}
\int_{[-\pi,\pi]^d} e^{-ix\cdot\frac{\xi}{\ep}}f(\xi)\,d\xi \, ,
\qquad f\in L^2(\T^d)
\end{equation}
with inverse ${\mathscr F}_\ep^{-1}:\ell^2\left((\ep{\mathbb
Z})^d\right)\to L^2\left(\T^d\right)$,
\begin{equation} \label{Fou-1}
({\mathscr F}_\ep^{-1}v)(\xi) :=
\frac{1}{\sqrt{2\pi}^d}\sum_{x\in\disk} e^{ix\cdot\frac{\xi}{\ep}}v(x),
\qquad v\in \ell^2\left(\disk\right)\: ,
\end{equation}
where $x\cdot y := \skp{x}{y} := \sum_{j=1}^d x_j y_j$ for $x,y\in \R^d$ and points in $\T^d$ are identified with
points in $[-\pi, \pi]^d$..
Then ${\mathscr F}_\ep$ is an isometry, i.e.,
\begin{equation}
\skpd{v}{u} = \skpT{{\mathscr F}_\ep^{-1}v}{{\mathscr F}_\ep^{-1}u}\:, \qquad   u,v\in
\ell^2\left(\disk\right)\label{dnachT}\; .
\end{equation}

On $L^2(\R^d)$ we denote by
$\skpR{f}{g} := \int_{\R^d}\overline{f}(\xi)g(\xi)\,d\xi$ the
standard scalar product and we introduce the $\ep$-scaled Fourier transform
\begin{eqnarray}\label{F}
(F_\ep^{-1}f)(\xi) &:=&(\ep \sqrt{2\pi})^{-d} \int_{\R^d} e^{\frac{i}{\ep}\xi\cdot x} f(x) \, dx \\
(F_\ep u)(x)  &:=& (\sqrt{ 2\pi})^{-d}\int_{\R^d} e^{-\frac{i}{\ep}\xi\cdot x} u(\xi)\, d\xi \, , \nonumber
\end{eqnarray}
where compared to the usual Fourier transform the roles of $x$ and $\xi$  are interchanged.
We notice that for any $f,g\in L^2(\R^d)$
\begin{equation}\label{parsevalkont}
\left\langle
F_\ep^{-1} f \,|\, F_\ep^{-1}g \right\rangle_{L^2(\R_\xi^d)} = \ep^{-d}
  \left\langle f \,|\, g \right\rangle_{L^2(\R_x^d)}\; .
\end{equation}
We write $\langle x \rangle := \sqrt{ 1 + |x|^2}$ for $x\in\R^d$.\\

We introduce the symbol-spaces $S(m)(\R^d\times \T^d)$ and $S_\delta^r(m)(\R^d\times\T^d)$ depending on
the small parameter $\ep\in (0,1]$ following Dimassi-Sj\"ostrand \cite{dima}.
A corresponding symbolic calculus is introduced in Appendix \ref{symbols}.

\begin{Def}\label{symbolspaces}
\begin{enumerate}
\item
A function $m: \R^d\times \T^d \rightarrow [0,\infty)$ is called an order function,
if there exist constants
$C_0, N_1>0$ such that
\[ m(x,\xi) \leq C_0\langle x-y \rangle^{N_1} m(y,\eta) \, , \qquad x,y\in\R^d,\, \xi,\eta\in\T^d \, .\]
\item
For $\delta\in[0,1]$, the space
$S^k_\delta(m)\left(\R^d\times \T^d\right)$ consists of
functions $a(x,\xi;\ep)$ on $\R^d\times\T^d\times (0,1]$, such that there exists constants $C_{\alpha,\beta} >0$ such that for all $\ep\in(0,1], (x,\xi)\in \R^d\times\T^d$
\begin{equation}\label{symbolbed}
|\partial^{\alpha}_x\partial^\beta_\xi a(x,\xi;\ep)| \leq C_{\alpha,\beta}
m(x,\xi)\ep^{k - \delta (|\alpha|+|\beta|)}\, .
\end{equation}
The best constants $C_{\alpha,\beta}$ in \eqref{symbolbed} are denoted by $ \|a\|_{\alpha,\beta}$ and
endow $S_\delta^k(m)(\R^d\times\T^d)$ with a Fr\'echet-topology.
\item
Let $a_j\in S_\delta^{k_j}(m), k_j\nearrow \infty$, then we write $a\sim
\sum_{j=0}^\infty a_j$ if
$a - \sum_{j=0}^Na_j \in S_\delta^{k_{N+1}}(m)$ for every $N\in \N$.
\end{enumerate}
\end{Def}

With $a\in S_\delta^k(m)\left(\R^d\times \T^d\right)$ we associate a
pseudo-differential operator $\Op_\ep^{\T}(a)$ formally given by
\begin{equation}\label{psdo2}
\Op_\ep^{\T}(a)\, v(x) := (2\pi)^{-d} \sum_{y\in\disk}\int_{[-\pi,\pi]^d} e^{\frac{i}{\ep}(y-x)\xi}
a(x,\xi;\ep)v(y) \, d\xi \, .
\end{equation}

We show in  Appendix \ref{symbols}, Lemma \ref{opastetig} that $Op_\ep^{\T^d} (a)$ is continuous
on
\begin{equation}\label{diskschwarz}
s\left(\disk\right):=\left\{ \left.u:\disk \rightarrow \C \; \right| \;
\|u\|_\alpha
:= \sup_{x\in\disk}
 \left| x^{\alpha} u (x)\right| < \infty ,\,
\alpha\in\N^d\right\}\, .
\end{equation}
equipped with the Fr\'echet-topology associated to the family of seminorms $\|\cdot\|_\alpha$. By
standard arguments, $s(\disk)$ is dense in $\ell^2(\disk)$.

For $a\in  S^r_\delta(1)\left(\R^d\times\T^d\right)$ with  $0\leq \delta \leq \frac{1}{2}$,
a version of the Calderon-Vaillancourt Theorem holds (Proposition \ref{cald}).
More precisely, there exists a constant $M>0$ such that for any $u\in s\left(\disk\right)$ and any $\ep\in(0,1]$
\[ \| \Op_\ep^{\T}(a) u \|_{\ell^2(\disk)} \leq M \ep^r \|u\|_{\ell^2(\disk)} \; .\]

Defining the $\#$-product
\[ \#:\, {\mathscr C}_0^\infty \left(\R^d\times\T^d\right)\times{\mathscr C}_0^\infty
\left(\R^d\times\T^d\right) \ni
(a,b)\longmapsto a\# b\in{\mathscr C}^\infty \left(\R^d\times\T^d\right) \]
by
\[
(a\# b) (x,\xi;\ep) :=
\left(e^{-i\ep D_y\cdot D_\xi}a(x,\xi;\ep)b(y,\eta;\ep)\right)|_{\natop{y=x}{\eta=\xi}},
\]
Corollary \ref{a1a2} tells us that
this product has a bilinear continuous extension to symbol spaces:
\[\#:\, S_{\delta_1}^{r_1}(m_1)\left(\R^d\times\T^d\right)\times
S_{\delta_2}^{r_2}(m_2)\left(\R^d\times\T^d\right)\rightarrow
S_\delta^{r_1+r_2}(m_1m_2)\left(\R^d\times\T^d\right)\]
for all
$\delta_k\in[0,\frac{1}{2}], k=1, 2$ and all order functions $m_1, m_2$,
where $\delta:=\max\{\delta_1, \delta_2\}$.
Furthermore, for $\delta_j<\frac{1}{2}$, $j=1,2$
it has the expansion
\[
(a\# b) \sim \sum_{j=0}^\infty (a\# b)_j \quad\text{with} \quad (a\# b)_j = \sum_{|\alpha|=j}
\frac{(i\ep)^{|\alpha|}}{|\alpha|!}\left(\partial_\xi^\alpha
a\right)
\left(\partial_x^\alpha b\right)
\]
in $S_\delta^{r_1+r_2}(m_1m_2)\left(\R^d\times\T^d\right)$ for all $a,b\in
S_{\delta_j}^{r_j}(m_j)\left(\R^d\times\T^d\right),\, j=1,2$ (in the sense of Definition \ref{symbolspaces},(e) with
$k_j=r_1+ r_2 + j(1-2\delta)$).\\
We recall that the $\#$-product reflects the composition of operators, i.e.
\[ \left(\Op_\ep^{\T}(a)\right)\circ \left(\Op_\ep^{\T}(b)\right) = \Op_\ep^{\T}(a\# b)\; .
 \]

Let $t$ denote the function defined in \eqref{talsexp}, then $t\in S_0^0(1)$. A straightforward calculation
gives $\Op_\ep^{\T}(e^{-\frac{i}{\ep}\gamma\cdot\xi})= \tau_\gamma$ and thus
\[ \Op_\ep^{\T}(t) = T_\ep\; .\]

\begin{rem}\label{torusfunktion}
Any function $f\in\Ce_0^\infty(\R^d_\xi)$, which is supported in $(-\pi,\pi)^d$,
admits a unique $\Ce^\infty$ periodic continuation to $\R^d$. Thus any such $f$ can be considered as a function on the
torus $\T^d$. We shall denote this function on $\T^d$ by $\tilde{f}$.
\end{rem}

Let $k\in {\mathscr C}_0^\infty\left(\R^d\right)$ be a cut-off function on $\R^d$ such that
$k(\xi)=1$ for $|\xi|\leq 2$ and $\supp k \subset (-\pi,\pi)^d$.
Then the truncated quadratic approximation of $t$ given by
\begin{equation}\label{tpiq}
t_{\pi,q}(x,\xi):= \left(\skp{\xi}{B(x)\xi}+ \ep t_1(x,0)\right)\,k(\xi) \, ,\qquad \xi\in \R^d,\, x\in\R^d,
\end{equation}
defines a function $\tilde{t}_{\pi,q}\in S_0^0(1)(\R^d\times \T^d)$ (with the notation of Remark \ref{torusfunktion}).
The associated bounded operator on the lattice (see \eqref{psdo2})
is denoted by $\Op_\ep^{\T}(\tilde{t}_{\pi,q}) =: T_{\ep,q}$.

Moreover we define for a potential well $x_j$ of $V_0$ in the sense of Hypothesis \ref{hypsatz}
\begin{equation}\label{tpiqj}
\tilde{t}_{\pi,q,j}(\xi):= \tilde{t}_{\pi,q}(x_j,\xi) \quad\text{and} \quad T_{\ep,q,j}:=\Op_\ep^{\T}(\tilde{t}_{\pi,q,j}) \; .
\end{equation}
To compare $H_\ep$ with an harmonic oscillator on
$L^2\left(\R^d\right)$,
we associate to $t$
(considered as an element of $S_0^0(1)(\R^{2d})$ in the sense of Definition \ref{defpsdokont}) the
translation operator on $L^2(\R^d)$
\begin{equation}\label{T}
\hat{T}_\ep:= \Op_\ep (t) = \sum_{\gamma\in\disk} a_\gamma (x, \ep) \tau_\gamma\; ,\qquad x\in\R^d\, ,
\end{equation}
(see \eqref{opkont}),
and we define the
associated Hamilton operator $\hat{H}_\ep$ on
$\De (\hat{H}_\ep) = \left\{\hat{V}_\ep u \in L^2\left(\R^d\right)\right\}$ as
\begin{equation}\label{Hepdach}
\hat{H}_\ep u(x) := \sum_{\gamma\in\disk} a_\gamma(x, \ep) u(x+\gamma) + \hat{V}_\ep (x) u(x)\, , \qquad
u\in \De(\hat{H}_\ep) \, .
\end{equation}
Setting $t_q(x,\xi) := \skp{\xi}{B(x)\xi} + \ep t_1(x,0)$ on $\R^d\times\R^d$,
we have
\begin{equation}\label{tq1}
\hat{T}_{q} := \Op_\ep (t_q) = -\ep^2\sum_{\nu,\mu =1}^dB_{\nu\mu}(x)\partial_\nu\partial_\mu + \ep t_1(x,0)\, .
\end{equation}
For $x_j\in\R^d$ as above we set
\begin{equation}\label{tqj}
t_{q,j}(\xi):= t_q(x_j,\xi)\, , \quad(\xi\in\R^d)\quad\text{and} \quad \Op_\ep (t_{q,j}) =: \hat{T}_{q,j}\; .
\end{equation}

\begin{rem}\label{Gxnull}
We denote by $\mathscr{G}_{x_0}= \disk + x_0$
the $\ep$-scaled lattice, shifted to the point $x_0\in\R^d$.
Then $x+\gamma\in\mathscr{G}_{x_0}$ for any
$x\in {\mathscr G}_{x_0}, x_0\in\R^d$ and $\gamma\in\disk$.
If $\id_{\mathscr{G}_{x_0}}$ is defined as the restriction map to the
lattice $\mathscr{G}_{x_0}$, it follows that $\tau_\gamma$ commutes with $\id_{\mathscr{G}_{x_0}}$.
Then $H_\ep = \hat{H}_\ep \id_{\mathscr{G}_0}$ and $H_{\ep,x_0}:= \hat{H}_\ep \id_{\mathscr{G}_{x_0}}$
defines a natural realization of $H_\ep$ on $\ell^2({\mathscr G}_{x_0})$.
\end{rem}

By Hypothesis \ref{hypsatz}, at a potential well $x_j$,
for $|x-x_j|\to 0$, the potential energy $\hat{V}_{\ep}$ has
 the expansion
\begin{align}
\hat{V}_{\ep}(x) &= V^j_{\ep}(x) + \ep\, O(|x-x_j|) + O(|x-x_j|^3)
+ R_2(x,\ep) \nonumber \\
 \text{where}\quad V^j_{\ep}(x) &:= V_0^j(x) + \ep\, V_1(x_j)\, ,\quad\text{and}\quad
 V_0^j(x) := \skp{(x-x_j)}{{\tilde A}^j(x-x_j)}\; .\label{Vdj}
\end{align}

\begin{rem}\label{remeig}
The operators $K_j$ defined in \eqref{Kj} are harmonic oscillators with the additional additive constant $V_1(x_j) + t_1(x_j,0)$.
Denoting by $(\omega^j_\nu)^2$ for $\omega^j_\nu>0$ the eigenvalues of the matrix $A^j$, the
eigenvalues of the operator $K_j$ are given by
\begin{equation}\label{sigmaKj}
\sigma(K_j) = \left\{\left. e_{\alpha,j} = \sum_{\nu =1}^d\left(\omega_\nu^j(2\alpha_\nu +1)\right) +
 V_1(x_j) + t_1(x_j,0)\; \right|\;
\alpha \in \N^d\right\}\:.
\end{equation}
The spectrum $\sigma (K)$ of $K$ is the union
$\sigma (K) = \bigcup_{j=1}^m \sigma (K_j)$ of the spectra
$\sigma (K_j)$ for all $j$.

The normalized eigenfunctions of the operators $K_j$ associated to an eigenvalue $e_{\alpha,j}$ are given by
\begin{equation}\label{gnkj}
g_{\alpha ,K_j}(x)= h_{\alpha}(x) e^{-\varphi_0^j(x)}\:, \qquad \alpha=(\alpha_1,\ldots ,
\alpha_d)\in\N_0^d\, ,
\end{equation}
where
\begin{equation}\label{prodhalpha}
h_{\alpha}(x)=h_{\alpha_1}(\langle x, y^j_1\rangle)\cdot h_{\alpha_2}(\langle x, y^j_2\rangle)\cdot\ldots\cdot
h_{\alpha_d}(\langle x, y^j_d\rangle)
\end{equation}
( $y_\nu^j\in{\mathbb R}^d$, ($\nu=1,\ldots, d$) denotes an orthonormal basis in $\R^d$
of eigenvectors of $A^j$),
and each $h_{\alpha_\nu}$ is a one-dimensional Hermite polynomial
\begin{equation}\label{hermitepol}
 h_k(t) = \frac{(-1)^k}{\sqrt{2^k k!} \pi^{\frac{1}{4}}}
e^{t^2} \left(\frac{d}{dt}\right)^k
e^{-t^2}
\end{equation}
with $k=\alpha_\nu$.
We assume $h_\alpha$ to be normalized in the sense that $\|g_{\alpha,K_j}\|_{L^2}=1$.
The phase function in \eqref{gnkj} is given by
\begin{equation}\label{varphi0}
\varphi_0^j(x) := \frac{1}{2}\sum_{\nu =1}^d\omega_\nu^j\skp{x}{y_\nu^j}^2 \, .
\end{equation}
\end{rem}

\section{Localization estimates}\label{beweis}

The starting point of the proof lies in
choosing a partition of unity.
This permits us to treat separately the neighborhoods of the minima and the
region outside of these neighborhoods.

By standard arguments, there exists $\chi\in {\mathscr C}_0^{\infty}\left(\R^d\right)$ such that
\begin{enumerate}
\item $0\leq \chi \leq 1$,
\item $\chi (x) = 1 $ if $|x|\leq 1$ and $\chi (x) = 0 $ if $|x| \geq 2$,
\item $\sqrt{\id -\chi^2}\in {\mathscr C}^{\infty}\left(\R^d\right)$.
\end{enumerate}
We define for $s>0$ functions which localize in $\ep^{s}$-scaled neighborhoods of the minima $x_j,\, 1\leq j\leq m$, by
\begin{equation} \label{defchi2}
\chi_{j,\ep,s}(x) := \chi \left(\ep^{-s}(x-x_j)\right) \; ,
\qquad x\in\R^d \, .
\end{equation}
For $\ep$ sufficiently small, $\supp \chi_{j,\ep,s} \cap \supp \chi_{k,\ep,s} = \emptyset$ for $k\neq j$ and thus by (c)
\[ \chi_{0,\ep,s} := \sqrt{\id-\sum_{j=1}^m \chi_{j,\ep,s}^2} \in {\mathscr C}^{\infty}\left(\R^d\right)
\quad\text{and}\quad
 \sum_{j=0}^m \chi_{j,\ep,s}^2 = \id\; .  \]
Furthermore we set for $j=0,1,\ldots, m$
\begin{equation}\label{defchi}
\chi_{j,\ep}:= \chi_{j,\ep,\frac{2}{5}}\, .
\end{equation}
Using this partition of unity, we obtain modulo $O\left(\ep^2\right)$ for $1\leq j \leq m$ with the notation
$V_1^j(x):= V_1(x_j)$
(using \eqref{Veps} and \eqref{Vdj})
\begin{align}\label{chiabV}
\left\| \chi_{j,\ep} \left(\hat{V}_{\ep}- V_{\ep}^j\right) \chi_{j,\ep}\right\|_{\infty} &=
\left\|\chi_{j,\ep} \left(\left(V_0-V_0^j\right)+\ep \left(V_1 - V_1^j\right)\right)\chi_{j,\ep}
\right\|_{\infty}\\
&\leq \sup_{x\in\supp (\chi_{j,\ep})}\left|\left(V_0-V_0^j\right)(x)\right|
 + \ep \left|\left(V_1-V_1^j\right)(x)\right|\nonumber\\
 &= O\left(\ep^{\frac{6}{5}}\right) \, ,\nonumber
 \end{align}
where the last estimate follows from $(V_0(x)- V_0^j(x))=O\left(|x-x_j|^3\right)$ and
$(V_1-V_1^j)(x) = O(|x-x_j|)$ as $x\to 0$ and from
$|x-x_j| = O\left(\ep^{\frac{2}{5}}\right)$ for
$x\in\supp (\chi_{j,\ep})$.
We shall now simultaneously localize $T_\ep$ around $\xi=0$ and $x=x_j$, which gives the main
contribution to the low-lying spectrum.
To this end we define a partition of unity by
\begin{equation}\label{phiaufR2}
\phi_{0,\ep,s}(\xi) := \chi (\ep^{-s}\xi),\qquad \xi\in\R^d
\end{equation}
and $\phi_{1,\ep,s} :=\sqrt{\id - \phi_{0,\ep,s}^2}$.
To $\phi_{0,\ep,s}\in\Ce_0^\infty(\R^d)$ we associate $\tilde{\phi}_{0,\ep,s}\in\Ce^\infty_0(\T^d)$
(see Remark \ref{torusfunktion}).
Then $\tilde{\phi}_{1,\ep,s}(\xi):=\sqrt{\id-\tilde{\phi}_{0,\ep,s}^2}\in{\mathscr C}^{\infty}(\T^d)$
satisfies $\tilde{\phi}_{0,\ep,s}^{2} + \tilde{\phi}_{1,\ep,s}^{2} = \id$.
The functions $\tilde{\phi}_{k,\ep,s}$ can be considered as elements of
$S^0_{\frac{2}{5}}(1)\left(\R^d\times\T^d\right)$ with
associated operator $\Op_\ep^{\T}(\tilde{\phi}_{k,\ep,s})$.
As above we set
\begin{equation}\label{phiaufR}
\phi_{k,\ep}:= \phi_{k,\ep,\frac{2}{5}}\quad\text{and}\quad \tilde{\phi}_{k,\ep}:= \tilde{\phi}_{k,\ep,\frac{2}{5}}\, ,\quad k=0,1\; .
\end{equation}

\begin{prop}\label{micnorm}
Let $T_{\ep}$ be a translation operator on the lattice $\disk$ as described in Hypothesis
\ref{hypsatz} with the symbol $t$ and let
$T_{\ep,q,j}$ denote the quadratic approximation of $T_\ep$, associated to the symbol
$t_{\pi, q,j}$ defined in (\ref{tpiqj}).
Let $\chi_{j,\ep},\, 1\leq j \leq m,$ and $\tilde{\phi}_{0,\ep}$ be the cut-off-functions defined in (\ref{defchi})
and (\ref{phiaufR}) respectively.
Then
\begin{equation}\label{mikroest}
\left\| P  \right\| = O\left(\ep^{\frac{6}{5}}\right) \, , \quad\text{where}\quad
P:=  \chi_{j,\ep}\Op_\ep^{\T}(\tilde{\phi}_{0,\ep})(T_{\ep} - T_{\ep,q,j})
\Op_\ep^{T^d}(\tilde{\phi}_{0,\ep})\chi_{j,\ep} \, .
\end{equation}
\end{prop}

\begin{proof}

By Proposition \ref{cald}, we only need to show that $p \in S^{\frac{6}{5}}_{\delta}(1)$ for some
$0\leq \delta\leq \frac{1}{2}$, where
$P=\Op_\ep^{\T}(p)$.
First we remark that for two symbols $a,b\in S^r_\delta (m), \, \delta < \frac{1}{2}$, where $b$ has compact
support, and a function
$\psi\in {\mathscr C}_0^\infty (\R^d\times \T^d)$ with $\psi |_{\supp b} =1$, we have by Corollary \ref{a1a2}
\begin{equation}\label{micnorm1}
a\# b (x,\xi,\ep) = a\psi \# b (x,\xi,\ep) + O\left(\ep^\infty \right) \, .
\end{equation}
Now choose cut-off-functions $\hat{\phi}_{0,\ep}(\xi)$ and ${\hat \chi}_{j,\ep}$ constructed as above from
$\hat{\chi}\in{\mathscr C}^\infty_0\left(\R^n\right)$ with $\hat{\chi}=1$ for $|x|\leq 2$ and $\hat{\chi}=0$ for
$|x|\geq 3$.
By Lemma \ref{AB} and \eqref{micnorm1}
it suffices to show that $\hat{p} \in S^{\frac{6}{5}}_{\frac{2}{5}}(1)$, where
\[ \hat{p}(x,\xi;\ep) := (\chi_{j,\ep}\#\tilde{\phi}_{0,\ep}\# (t-t_{\pi,q,j})\hat{\phi}_{0,\ep}{\hat\chi}_{j,\ep}\#
\tilde{\phi}_{0,\ep}\#\chi_{j,\ep}) (x,\xi;\ep) \; .\]

We first determine the symbol class of $(t-t_{\pi,q,j})\hat{\phi}_{0,\ep}{\hat \chi}_{j,\ep}$.
Let $\alpha,\beta\in\N^d$, then for $\alpha_1+\alpha_2 = \alpha$ and
$\beta_1+\beta_2 = \beta $
\begin{multline}\label{micnorm2}
\left| \partial_x^\alpha \partial_\xi^\beta (t-t_{\pi,q,j})\hat{\phi}_{0,\ep}{\hat \chi}_{j,\ep} (x,\xi, \ep)
\right| \\
=
\left| \sum_ {\alpha_1,\alpha_2,\beta_1,\beta_2}\!\!\!
\left( \partial_x^{\alpha_1}
\partial_\xi^{\beta_1} (t-t_{\pi,q,j})(x,\xi, \ep)\right)
\left(\partial_\xi^{\beta_2}\hat{\phi}_{0,\ep}(\xi)\right)
\left(\partial_x^{\alpha_2}\hat{\chi}_{j,\ep}(x)\right)
\right| \, .
\end{multline}
Writing $t-t_{\pi,q,j} = t-t_{\pi,q}+ t_{\pi,q} - t_{\pi,q,j}$, we
use that by the definition of $t_1$, Hypothesis
\ref{hypsatz},(a),(i) and Remark \ref{talskin},(a) for each
$x\in\R^d$
\begin{align}\label{ttpiqjdiff}
(t-t_{\pi,q})(x,\xi,\ep) &= O(|\xi|^4) + \ep O(|\xi|) + O(\ep^2)\qquad\text{and}\\
(t_{\pi,q} - t_{\pi,q,j})(x,\xi,\ep) &= \skp{\xi}{(B(x)-
B(x_j))\xi} + \ep (t_1(x,0) - t_1(x_j,0)) = O(|\xi|^2)O(|x|) + \ep O(|x|)\, .\nonumber
\end{align}
The scaling in the definition of the cut-off-functions yields
$|x-x_j|=O\left(\ep^{\frac{2}{5}}\right) = |\xi|$, therefore by \eqref{ttpiqjdiff}
\begin{equation}
\sup_{|\xi|\in\supp (\hat{\phi}_{0,\ep})}\sup_{|x|\in\supp ({\hat\chi}_{j,\ep})}
\left(\partial_x^{\alpha_1}\partial_\xi^{\beta_1}(t-t_{\pi,q,j})(x,\xi;\ep)\right)
 \leq C \ep^{\frac{6}{5}-|\beta_1|\frac{2}{5}-|\alpha_1|\frac{2}{5}}\, .\label{micnorm3}
\end{equation}
By construction $\hat{\phi}_{0,\ep}, \hat{\chi}_{j,\ep}\in S^0_{\frac{2}{5}}(1)$, thus inserting \eqref{micnorm3} in
\eqref{micnorm2} shows
\[ \left| \partial_x^\alpha \partial_\xi^\beta (t-t_{\pi,q,j})\hat{\phi}_{0,\ep}{\hat \chi}_{j,\ep} (x,\xi;\ep)\right|
\leq
C_{\alpha,\beta}\ep^{\frac{6}{5}-\frac{2}{5}(|\alpha|+|\beta|)} \]
and therefore $(t-t_{\pi,q,j})\hat{\phi}_{0,\ep}{\hat \chi}_{j,\ep} \in S_{\frac{2}{5}}^{\frac{6}{5}}(1)$.
The cut-off-functions $\chi_{j,\ep}$ and $\tilde{\phi}_{0,\ep}$ are both elements of $S_{\frac{2}{5}}^0(1)$, thus by
Corollary \ref{a1a2} we get $p\in  S_{\frac{2}{5}}^{\frac{6}{5}}(1)(\R^d\times \T^d)$.
The estimate of the norm of the associated operator in $\ell^2\left(\disk\right)$ follows by use of
Proposition \ref{cald}.
\end{proof}

\begin{rem}
Using the symbolic calculus introduced in Dimassi-Sj\"ostrand \cite{dima}, in particular Proposition 7.7, Theorem 7.9 and
Theorem 7.11, it is possible to show
by similar considerations as in the lattice case that for $\hat{T}, \hat{T}_{qj}$ defined in
\eqref{T} and \eqref{tqj} respectively, with the cut-off functions $\chi_{j,\ep}, \phi_{k,\ep}$ defined in
\eqref{defchi} and \eqref{phiaufR}, one has the norm estimate
\begin{equation}\label{OpVald}
\|\chi_{j,\ep}(x)\tilde{\phi}_{0,\ep}(\ep D)(T_{\ep} - T_{\ep qj})\tilde{\phi}_{0,\ep}(\ep D)\chi_{j,\ep}(x)
\|_{\infty} = O(\ep^{\frac{6}{5}})\; .
\end{equation}
\end{rem}

\eqref{mikroest} suggests to define (see \eqref{Vdj})
\begin{equation}\label{tildeHj}
\hat{H}^j := \hat{T}_{q,j} + V_0^j + \ep\,V_1(x_j)  = \hat{T}_{q,j} + V_\ep^j
\end{equation}
as an approximating operator of $\hat{H}_{\ep}$ and $H_\ep$ respectively on $L^2 (\R^d)$.
By means of the unitary transformation $Uf(x):= \sqrt{|\det B_j^{-\frac{1}{2}}|} f(B_j^{-\frac{1}{2}}x)$, the operator $\hat{H}^j$ is unitarily equivalent to
\begin{equation}\label{Hjj}
H^j := -\ep^2\Delta  + \skp{(x-x_j)}{A^j(x-x_j)} + \ep\,( V_1(x_j) + t_1(x_j,0)) = U^{-1}\hat{H}^j U\, ,
\end{equation}
where $A^j, B_j$ are defined as in Theorem \ref{satz1}. Furthermore, by scaling,
$H^j$ is unitarily equivalent to $\ep K_j$.
Thus the spectrum of $H^j$ and $\hat{H}_j$ is given by
$\ep\:\sigma(K_j)$. The eigenfunctions of $H^j$ and $\hat{H}^j$ are
\begin{equation}\label{gnj}
g'_{\alpha j}(x) = \ep^{-\frac{d}{4}}h_{\alpha}\left(\tfrac{x-x_j}{\sqrt{\ep}}\right)
e^{-\varphi_0^j(\frac{x-x_j}{\sqrt{\ep}})}\quad\text{and} \quad
g_{\alpha j}:= U \, g'_{\alpha j}\quad\text{respectively}\;  \: .
\end{equation}

We will show now that modulo terms of order $\ep^{\frac{6}{5}}$ one can decompose $H_\ep$ with respect to
the partition of unity introduced above into a sum of Dirichlet operators. This is a generalization of
the IMS-localization formula for Schr\"odinger operators described for example in
Cycon-Froese-Kirsch-Simon \cite{simon}.

\begin{Lem}\label{1zer}
Let $H_{\ep} = T_{\ep} + V_{\ep}$ satisfy Hypothesis \ref{hypsatz} and
denote by $V_\ep^j$ the quadratic approximation of $V_\ep$ defined in (\ref{Vdj}).

Let $\chi_{j,\ep}\: , 0\leq j \leq m$ and $\tilde{\phi}_{k,\ep}\: , k=0,1$ be given
by (\ref{defchi}) and (\ref{phiaufR}) respectively.
Then the following estimates hold in operator norm.
\begin{enumerate}
\item \[H_{\ep} = \sum_{j=0}^m \chi_{j,\ep}\,H_{\ep}\,\chi_{j,\ep} +
O\left(\ep^{\frac{6}{5}}\right)\: .\]
\item \[T_{\ep} + V_\ep^j = \Op_\ep^{\T}(\tilde{\phi}_{0,\ep})(T_{\ep} +
V_{\ep}^j)\Op_\ep^{\T}(\tilde{\phi}_{0,\ep}) +
\Op_\ep^{\T}(\tilde{\phi}_{1,\ep})(T_{\ep} + V^j_{\ep})\Op_\ep^{\T}(\tilde{\phi}_{1,\ep}) +
O\left(\ep^{\frac{6}{5}}\right)\,.\]
\end{enumerate}
\end{Lem}

\begin{proof}
(a):\\
$H_{\ep}$ can be written as as
\begin{equation}\label{Hdauf}
 H_{\ep} = \frac{1}{2}\sum_{j=0}^m\chi_{j,\ep}^2\,H_{\ep} + \frac{1}{2}H_{\ep}\,
\sum_{j=0}^m\chi_{j,\ep}^2 =
\sum_{j=0}^m\chi_{j,\ep}\,H_{\ep}\,\chi_{j,\ep} +
\frac{1}{2}\sum_{j=0}^m\left[\chi_{j,\ep},[\chi_{j,\ep},H_{\ep}]\right] \, ,
\end{equation}
therefore we have to estimate the double
commutators on the right hand side of \eqref{Hdauf}.
Since $t\in S_0^0(1)$ and $\chi_{j,\ep}\in S_{\frac{2}{5}}^0(1)$,
$j=0,\ldots m$,
it follows at once from Lemma \ref{kommut} that $[\chi_{j,\ep},[\chi_{j,\ep}, t]_\#]_\# \in
S_{\frac{2}{5}}^{\frac{6}{5}} (1)$, which leads to (a) by Proposition \ref{cald}.

(b):\\
The arguments are quite similar to (a), but we need to consider the expansions for the symbolic double
commutator, since
the quadratic potential $V_\ep^j$ is not bounded, but $V_\ep^j \in S_0^0(|x|^2)$. Thus the general result
on the symbol class of the double commutator given in Lemma \ref{kommut} does not allow to use
Proposition \ref{cald} directly.
By Lemma \ref{kommut}, the double commutator
in the symbolic calculus with $\alpha,\alpha_1,\alpha_2\in\N^d$ for
$k=0,1$ can be written as
\begin{multline*}
[\tilde{\phi}_{k,\ep}(\xi),[\tilde{\phi}_{k,\ep}(\xi), (t + V^j_\ep)(x,\xi)]_\#]_\# \\
=\sum_{|\alpha|= 2}\!\!(i\ep)^{|\alpha|} \left(\partial_x^\alpha (t + V_\ep^j)\right)\!\!(x,\xi)\!\!\!\!\!\!\!
\sum_{\alpha_1+\alpha_2=\alpha}
\left(\partial_\xi^{\alpha_1}\tilde{\phi}_{k,\ep}\right)
\left(\partial_\xi^{\alpha_2}\tilde{\phi}_{k,\ep}\right)(\xi) + R_3\, .
\end{multline*}
Now we use that $t\in S_0^0(1)$ and $\tilde{\phi}_{k,\ep}\in S_{\frac{2}{5}}^0(1)$ and furthermore that
the second derivative of the quadratic term $V_\ep^j$ is constant. Thus all the summands are
bounded, of order $\ep^{2-\frac{4}{5}}$ and the $\ep$-order in lowered by $\frac{2}{5}$ with each
differentiation, i.e., they are elements of $S_{\frac{2}{5}}^{\frac{6}{5}} (1)$.
By Lemma \ref{kommut}, the remainder $R_3$ depends linearly on a finite number of derivatives
$\partial_x^\beta (h+V_\ep^j)$ with $|\beta|\geq 3 $ (which is bounded) and
$\left(\partial_\xi^{\beta_1}\tilde{\phi}_{k,\ep}\right)\left(\partial_\xi^{\beta_2}\tilde{\phi}_{k,\ep}\right)$
with $|\beta_1| + |\beta_2|\geq 3$. Thus it is an element of $S_{\frac{2}{5}}^{\frac{9}{5}}(1)$.
We therefore get
$[\tilde{\phi}_{k,\ep}(\xi),[\tilde{\phi}_{k,\ep}(\xi), (t + V^j_\ep)(x,\xi)]_\#]_\# \in S_{\frac{2}{5}}^{\frac{6}{5}}(1)$,
yielding
by Proposition \ref{cald} the stated norm estimate for the associated operator.
\end{proof}

We shall now restrict the eigenfunctions $g_{\alpha j}$ of $\hat{H}^j$ introduced in \eqref{gnj} to the lattice
$\disk$. We denote these restrictions by $g_{\alpha j}^\ep$ and we shall use them as approximate eigenfunctions
for $H_\ep$.

\begin{Lem}\label{kleinnorm}
Let $f, g$ denote eigenfunctions of $\hat{H}^j$ as defined in \eqref{gnj} and $f^\ep, g^\ep$ their
restriction to $\disk$. Then
\begin{equation}\label{gepdurchg}
\skpd{g^{\ep}}{f^{\ep}} =\ep^{-d}\left( \skpR{g}{f} + O(\sqrt{\ep})\right)\; .
\end{equation}
\end{Lem}

\begin{proof}
We use $\ep^d = \int_{[x,x+\ep)^d} dx$ to write
\begin{equation}\label{fgzerleg}
 \skpd{f^{\ep}}{g^{\ep}} = I_1 + I_2 + I_3\, ,
\end{equation}
where
\begin{align}\label{Ifg}
I_1 &=  \ep^{-d}\sum_{x\in\disk}\int_{[ x,
x+\ep[^d} (f(x) - f( y))\,g( x)\, dy \\
I_2 &=\ep^{-d}\sum_{x\in\disk}\int_{[ x,
x+\ep[^d}  f(y)\,(g(x) - g(y))\, dy \nonumber
\end{align}
and
\[ I_3 = \ep^{-d} \int_{\R^d}  f(y)g(y) \, dy = \ep^{-d}\skpR{f}{g} \: .
\]
It thus remains to show that $I_1$ and $I_2$ are of order $\ep^{-d + \frac{1}{2}}$.
By the scaling of $f$ and since $f=O(\ep^{-\frac{d}{4}})$
\begin{equation}\label{fdiff}
\sup_{x\in\disk}\sup_{y\in[x,x+\ep)^d} |f(x)-f(y)| \leq \ep \sup_{z\in\R^d} |\nabla f(z)|
\leq C \ep^{-\frac{d}{4}}\ep^{\frac{1}{2}}\; .
\end{equation}
Thus, setting $g(x) = \ep^{-\frac{d}{4}}\tilde{g}(\frac{x-x_j}{\sqrt{\ep}})$ for some $j\in\{ 1,\ldots m\}$,
we have by \eqref{Ifg} and \eqref{fdiff}
\begin{equation}\label{vterm1e}
|I_1| \leq C \ep^{\frac{1}{2} - \frac{d}{2}} \sum_{y\in\sqrt{\ep}\Z^d} \tilde{g}\left(y-\tfrac{x_j}{\sqrt{\ep}}\right) =
O\left(\ep^{\frac{1}{2}-d}\right)\; ,
\end{equation}
where in the last step we used that by the definition of the Riemann Integral
\begin{equation}\label{lemma2.10.2}
\lim_{\ep\to 0} \;\ep^{\frac{d}{2}}\sum_{y\in(\sqrt{\ep}\Z)^d} \left|\tilde{g}\left(y-\tfrac{x_j}{\sqrt{\ep}}\right)\right| =
\int_{\R^d} |\tilde{f}(u)|\, du\; ,
\end{equation}
which is a constant independent of $\ep$.
The estimates for $I_2$ are analogous.
\end{proof}

The functions $g_{\alpha j}$ defined in \eqref{gnj} are localized near the well $x_j$ for $j=1,\ldots,m$
and decrease exponentially fast.
We need the following localization estimates.

\begin{Lem}\label{lzweiab}
For $s<\frac{1}{2}$ let , $\chi_{j,\ep}$, $\chi_{j,\ep,s},\, 1\leq j \leq m$ and $\tilde{\phi}_{0,\ep,s}$, denote the
cut-off functions defined in \eqref{defchi}, \eqref{defchi2} and below \eqref{phiaufR2} respectively.
Let $g_{\alpha j}^{(\ep)}$ denote the eigenfunctions of the harmonic oscillator defined in
\eqref{gnj} (or their restriction to the lattice).
Then for $\ep \to 0:$
\ben
\item There exists a constant $C>0$ such that
\[ \left| \skpR{g_{\alpha j}}{\left(1-\chi^2_{j,\ep,s}\right)g_{\alpha j}}\right| =
O\left(e^{-C \ep^{2s-1}}\right)\; . \]
\item For all $N\in\N$
\[ \left| \skpT{\mathscr{F}_\ep^{-1}\left(\chi_{j,\ep}g_{\alpha j}^\ep\right)}{\phi_{1,\ep,s}^2
\mathscr{F}_\ep^{-1}\left(\chi_{j,\ep}g_{\alpha j}^\ep\right)} \right|
= O\left(\ep^{N}\right)\, . \]
\een
\end{Lem}

\begin{proof}

(a):\\
Estimating of $\id-\chi_{j,\ep,s}^2$ by $\id$ on its support gives
\[
\left| \skpR{g_{\alpha j}}{\left(1-\chi^2_{j,\ep,s}\right)g_{\alpha j}}\right|
\leq \int_{|x-x_j|\geq \ep^{s}}\left|g_{\alpha j}(x) \right|^2\, dx
\]
Using $g_{\alpha j}=O(\ep^{-\frac{d}{4}})$ and the
exponential decay of $g_{\alpha j}$,
the right hand side can be estimated from
above by
\begin{equation}\label{abgauss}
c \ep^{-\frac{d}{2}} \int_{|x-x_j|\geq \ep^{s}} p(|u|) e^{-C\frac{|x-x_j|^2}{\ep}}\, d|x-x_j| =
O\left(e^{-\frac{C}{2} \ep^{2s-1}}\right)
\end{equation}
for some $c,C>0$ and some polynomial $p$, proving (a).\\

(b):\\
To prove this statement, we sum by parts.
Setting $v:= \mathscr{F}_\ep^{-1}\left(\chi_{j,\ep}g_{\alpha j}^\ep\right)$ and
replacing the function $\tilde{\phi}_{1,\ep,s}$ on its support by $\id$, we get
\begin{equation}\label{strichb2}
\left| \skpT{v}{\phi_{1,\ep,s}^2 \, v} \right|
\leq \int\limits_{ \natop{[-\pi,\pi]^d}{|\xi|\geq \ep^s}}
|v(\xi)|^2\,d\xi  \: .
\end{equation}
We now estimate $|v(\xi)|^2$.
By the definition (\ref{Fou-1}) of the inverse Fourier transform,
\begin{equation}\label{vFou}
v(\xi) = \frac{1}{\sqrt{2\pi}^d}
\sum_{y\in\disk}e^{\frac{i}{\ep}\xi \cdot y}
\chi_{j,\ep}(y) g^{\ep}_{\alpha j}(y)\; .
\end{equation}
To analyze $v$ and $\bar{v}$, we use
summation by parts and the discrete Laplace operator $\Delta_\ep$
\begin{equation}\label{diskretlap}
(\Delta_{\ep}f)(x) :=
 \left(\sum_{\nu=1}^d(\tau_{\ep e_\nu} + \tau_{-\ep e_\nu})-2d\right)
f(x)\: .
\end{equation}
The operator $\Delta_\ep$ is symmetric in $\ell^2(\disk)$, i.e.,
\begin{equation}\label{symm}
\skpd{f}{\Delta_{\ep}h} = \skpd{\Delta_{\ep}f}{h}\; , \qquad f,h\in \ell^2(\disk)\; .
\end{equation}
By \eqref{diskretlap} we have
\begin{equation}\label{deltaeps}
 e^{\pm\frac{i}{\ep}x\cdot\xi} =
 -\,\left(2d-2\sum_{\nu =1}^d\cos \left(\xi_\nu\right)\right)^{-1}
 \Delta_{\ep} e^{\pm\frac{i}{\ep}x\cdot\xi}\: .
\end{equation}
Combining \eqref{vFou}, \eqref{deltaeps} and \eqref{symm} for any $N\in\N$ leads to
\begin{equation}\label{partsum3}
 \sqrt{2\pi}^d \, v(\xi) =
 -\,\left(2d-2\sum_{\nu =1}^d\cos \left(\xi_\nu\right) \right)^{-N}\sum_{x\in\disk}
(\Delta_{\ep}^N\chi_{j,\ep} g_{\alpha j}^{\ep})(x)
 e^{\frac{i}{\ep}x\cdot\xi}\; .
\end{equation}
We shall estimate the first factor on the right hand side of \eqref{partsum3} for
$\xi\in\mathcal{M}_\ep := \{ \xi\in [-\pi,\pi]^d\, |\, |\xi| \geq \ep^s \}$.
From the inequality $\pi^2(1-\cos \xi_\nu )\geq \xi_\nu^2$ for $|\xi_\nu|\leq\pi$ it follows that
\[
\frac{1}{\sum_{\nu =1}^d(2-2\cos
(\xi_\nu ))}\leq\frac{\pi^2}{2\sum_{\nu =1}^d\xi_\nu^2}
=\frac{\pi^2}{2|\xi|^2}\, , \qquad \xi\in\mathcal{M}_\ep
\]
and therefore
\begin{equation}\label{abep12}
\left(\sum_{\nu =1}^d(2-2\cos (\xi_\nu ))\right)^{-N}\leq
\left(\frac{\pi^2}{2\ep^{2s}}\right)^N
= O\left(\ep^{-2Ns}\right)\, , \qquad \xi\in\mathcal{M}_\ep \: .
\end{equation}
To find an estimate for the remaining series on the right hand side of \eqref{partsum3},
we use the differentiability of the
functions $\chi_{j,\ep} g_{\alpha j}$. We set $u:= \ep^{\frac{d}{4}}\chi_{j,\ep}g_{\alpha j}$, then by the chain rule and the
scaling of $g_{\alpha j}$ and $\chi_{j,\ep}$
\[
\partial^2_\nu u(x) = O\left(\ep^{-1}\right)\, , \qquad x\in\R^d\; .
\]
Thus Taylor expansion gives
\begin{align}\label{partsum33}
 \Delta_{\ep}u(x) &=  \sum_{\nu =1}^d
(u(x+\ep\, e_\nu )-
u(x)) + (u(x-\ep\, e_\nu ) - u(x))
\\
&= \ep^2 \sum_{\nu=1}^d \int_0^1\left( \partial_\nu^2 u(x+ t\ep e_\nu) +
\partial_nu^2 u(x-t\ep e_\nu)\right) \, dt \nonumber\\
&= O(\ep)\; .\nonumber
\end{align}
Iterating \eqref{partsum33} gives
\begin{equation}\label{Ndisklap}
\Delta_\ep^N u(x) = O\left( \ep^{N}\right) \; .
\end{equation}
Inserting \eqref{Ndisklap} and \eqref{abep12} into \eqref{partsum3} gives
\begin{equation}\label{vnorm}
|v(\xi)|^2 = O\left(\ep^{-\frac{d}{2} + N(1-2s)}\right)\, ,\qquad \xi\in\mathcal{M}_\ep\, .
\end{equation}
Inserting \eqref{vnorm} into \eqref{strichb2} shows (b).
\end{proof}

In the following lemma we use the above results to analyze the difference of matrix elements
for $H_\ep$, $V^j_\ep$ and $T_\ep$ and their localized approximations in the case $s=\frac{2}{5}$.

\begin{Lem}\label{chiphi}
Let $H_\ep$ and $T_\ep$ be given as in Hypothesis \ref{hypsatz},
$V_\ep^j$ be given in \eqref{Vdj} and $\hat{T}_{q,j}$ in \eqref{tqj}.
Let $\tilde{\phi}_{0,\ep}$, $\phi_{0,\ep}$ and $\chi_{j,\ep},\, 1\leq j \leq m$, denote the
cut-off functions defined in \eqref{phiaufR} and \eqref{defchi} respectively.
Let $g_{\alpha j}^{(\ep)}$ denote the eigenfunctions of $\hat{H}^j$ defined in
\eqref{gnj} (or their restriction to the lattice).
Then for $\ep \to 0:$
\begin{enumerate}
\item
\begin{equation}\label{chigg}
\left|\skpd{g^{\ep}_{\alpha j}}{H_\ep g^{\ep}_{\beta l}} - \skpd{\chi_{j,\ep} g^{\ep}_{\alpha j}}
{H_{\ep} \chi_{j,\ep} g^{\ep}_{\beta l}}\right| = O\left(\ep^{\frac{6}{5}}\right)\; .
\end{equation}
\item There exists a constant $c>0$ such that
\[
\left|\skpR{g_{\alpha j}}{V^j_{\ep} g_{\beta l}} - \skpR{\chi_{j,\ep} g_{\alpha j}}
{V^j_{\ep} \chi_{j,\ep} g_{\beta l}}\right| = O\left(e^{-c\ep^{-\frac{1}{5}}}\right)\; .
\]
\item 
\begin{equation}\label{phistrichgg}
 \left| \skpd{\chi_{j,\ep} g^{\ep}_{\alpha j}}{T_\ep\,
 \chi_{j,\ep} g^{\ep}_{\beta l}} -
\skpd{\Op_\ep^{\T}(\tilde{\phi}_{0,\ep}) \chi_{j,\ep} g^{\ep}_{\alpha j})}
{T_\ep\, \Op_\ep^{\T}(\tilde{\phi}_{0,\ep}) \chi_{j,\ep} g^{\ep}_{\beta l}} \right|
= O\left(\ep^{\frac{6}{5}}\right)\; .
\end{equation}
\item
\[
\left| \skpR{g_{\alpha j}}{\hat{T}_{q,j}\, g_{\beta l}} -
\skpR{\Op_\ep (\phi_{0,\ep})\chi_{j,\ep} g_{\alpha j}}{\hat{T}_{q,j}\,
\Op_\ep(\phi_{0,\ep})\chi_{j,\ep} g_{\beta l}} \right| =
      O\left(\ep^{\frac{6}{5}}\right)\; .
\]
\end{enumerate}
\end{Lem}

\begin{proof}

(a):\\
By Lemma \ref{1zer}
\begin{eqnarray}\label{chigg3}
\lefteqn{\left|\skpd{g^{\ep}_{\alpha j}}{H_{\ep} g^{\ep}_{\beta j}} - \skpd{\chi_{j,\ep}
g^{\ep}_{\alpha j}}{H_{\ep} \chi_{j,\ep}
g^{\ep}_{\beta j}}\right|  }\hspace{1.5cm}\\
&& =\left|\sum_{k\neq j}\skpd{\chi_{k,\ep}(x) g_{\alpha j}^{\ep}}{(T_\ep + V_\ep)
\chi_{k,\ep}(x) g_{\beta j}^{\ep}}
\right| +  O\left(e^{\frac{6}{5}}\right).\nonumber
\end{eqnarray}
We consider the kinetic and potential term separately, starting with the
potential term $V_\ep$. By estimating $(\id-\chi_{j,\ep}^2)$ on its support by $1$ and
using $g_{\beta l}=O(\ep^{-\frac{d}{4}})$, we get for some $C>0$
\[
\left|\skpd{g_{\alpha j}^{\ep}}{(\id - \chi_{j,\ep}^2)V_{\ep} g_{\beta l}^{\ep}} \right|
\leq C \ep^{-\frac{d}{4}} \sum_{\natop{x\in\disk}{|x-x_j|\geq \ep^{\frac{2}{5}}}}\left|V_{\ep}(x)
g_{\alpha j}^{\ep}(x)\right|\; .
\]
$V_\ep$ is by Hypothesis \ref{hypsatz} polynomially bounded, thus  the right hand side is bounded from above by
\[
 C\ep^{-\frac{d}{2}}\sum_{|x-x_j|\geq \ep^{\frac{2}{5}}} |p(|x-x_j|)|e^{-c \frac{|x-x_j|^2}{\ep}}
\]
for some $c,C>0$ and some polynomial $p$. This yields for some $c>0$
\begin{equation}\label{lemma291}
\left|\skpd{g_{\alpha j}^{\ep}}{(\id - \chi_{j,\ep}^2)V_{\ep}g_{\beta l}^{\ep}} \right| =
O\left(e^{-c\ep^{-\frac{1}{5}}}\right)\; .
\end{equation}
The boundedness of $T_\ep$ together with Lemma \ref{lzweiab} yields
\begin{equation}\label{lemma292}
\left|\sum_{k\neq j}\skpd{\chi_{k,\ep}(x) g_{\alpha j}^{\ep}}{T_\ep \chi_{k,\ep}(x) g_{\beta j}^{\ep}}
\right|\leq C \sum_{k\neq j} \|\chi_{k,\ep}g_{\alpha j}^\ep\|_{\ell^2} = O\left(e^{-c\ep^{-\frac{1}{5}}}\right)
\end{equation}
for some $c>0$.
Inserting \eqref{lemma291} and \eqref{lemma292} in (\ref{chigg3}) shows the stated estimate.

(b):\\
This is analogue to the proof of Lemma \ref{lzweiab}, since $V^j_\ep$ just changes the
polynomial term in \eqref{abgauss}.

(c):\\
By Lemma \ref{1zer},
\begin{multline}\label{stricha}
 \left| \skpd{\chi_{j,\ep} g^{\ep}_{\alpha j}}{T_\ep\,
 \chi_{j,\ep} g^{\ep}_{\beta l}} -
\skpd{\Op_\ep^{\T}(\tilde{\phi}_{0,\ep}) \chi_{j,\ep} g^{\ep}_{\alpha j}}
{T_\ep\, \Op_\ep^{\T}(\tilde{\phi}_{0,\ep})\chi_{j,\ep} g^{\ep}_{\beta l}} \right|  \\
=  \left| \skpd{\Op_\ep^{\T}(\tilde{\phi}_{1,\ep}) \chi_{j,\ep} g^{\ep}_{\alpha j}}
{T_\ep\, \Op_\ep^{\T}(\tilde{\phi}_{1,\ep}) \chi_{j,\ep} g^{\ep}_{\beta l}}\right| +
O \left(\ep^{\frac{6}{5}}\right)
\end{multline}
Since by \eqref{AB}
$T_\ep \Op_\ep^{\T}(\tilde{\phi}_{1,\ep}) = \Op_\ep^{\T}(t \tilde{\phi}_{1,\ep})$, we have
by the isometry of $\mathscr{F}_\ep$
\begin{multline}\label{strichb}
\left| \skpd{\Op_\ep^{\T}(\tilde{\phi}_{1,\ep}) \chi_{j,\ep} g^{\ep}_{\alpha j}}
{T_\ep\, \Op_\ep^{\T}(\tilde{\phi}_{1,\ep}) \chi_{j,\ep} g^{\ep}_{\beta l}}\right| =
\left| \skpT{\tilde{\phi}_{1,\ep}{\mathscr F}_\ep^{-1}(\chi_{j,\ep} g^{\ep}_{\alpha j})}
{t\, \tilde{\phi}_{1,\ep} {\mathscr F}_\ep^{-1}(\chi_{j,\ep} g^{\ep}_{\beta l})} \right|\\
\leq
C \,\|\tilde{\phi}_{1,\ep}{\mathscr F}_\ep^{-1}(\chi_{j,\ep} g^{\ep}_{\alpha j})\|_{\T}\,
 \|\tilde{\phi}_{1,\ep} {\mathscr F}_\ep^{-1}(\chi_{j,\ep} g^{\ep}_{\beta l})\|_{\T}  = O\left(\ep^\infty\right)\; ,
\end{multline}
where the second estimate follows from the boundedness of $t$ and the last from Lemma \ref{lzweiab}, (b).

(d):\\
We set $P:=\Op_\ep\left(\phi_{0,\ep}\right)^2 \hat{T}_{q,j}$, with symbol
$p(\xi)= \phi_{0,\ep}^2(\xi) t_{q,j}(\xi)$. Then, using \eqref{parsevalkont},
\begin{multline}\label{TundP}
\left|  \skpR{g_{\alpha j}}{\hat{T}_{q,j} g_{\beta l}} - \skpR{g_{\alpha j}}{P g_{\beta l}} \right|
 =
\ep^d \left|
\skpR{F_\ep^{-1}g_{\alpha j}}{\phi_{1,\ep}^2 t_{q,j}\,F_\ep^{-1} g_{\beta l}}
\right|
 \\
 \leq \int_{|\xi|\geq \ep^{\frac{2}{5}}}
\left|
\left(F_\ep^{-1}g_{\alpha j}\right)(\xi)
t_{q,j}(\xi)
\left(F_\ep^{-1}g_{\beta l}\right)(\xi) \right| \, d\xi
\leq C e^{c\ep^{-\frac{1}{5}}}\; ,
\end{multline}
where we used that
\[ \left| \left(F_\ep^{-1}g_{\alpha j}\right)(\xi)\right| \leq C \ep^{-N} \left| q(\xi) e^{-c\frac{|\xi|^2}{\ep}}\right|
\]
for some $N\in\N, C,c >0$ and some polynomial $q(\xi)$.
Next observe that by \eqref{Hdauf}
\begin{equation}\label{Pzerleg}
P - \sum_{j=0}^m \chi_{j,\ep}P\chi_{j,\ep} = \frac{1}{2} \left[ \chi_{j,\ep},\,\left[ \chi_{j,\ep}, P\right]\right] =
O\left(\ep^{\frac{6}{5}}\right)\, ,
\end{equation}
since $p\in S^{\frac{4}{5}}_{\frac{2}{5}}(1)$ and $\chi_{j,\ep} \in S^{0}_{\frac{2}{5}}(1)$, using
PDO-calculus on $\R^d$, in particular the Theorem of Calderon and Vaillancourt (see \cite{dima}). Furthermore
\begin{equation}\label{zweiteab}
\sum_{k\neq j} \left| \skpR{g_{\alpha j}}{\chi_{k,\ep}P\chi_{k,\ep} g_{\beta l}}\right| \leq
\sum_{k\neq j} \| \chi_{k,\ep} g_{\alpha j}\|_{L^2} \, \|P\chi_{k,\ep} g_{\beta l}\|_{L^2}
= O\left(e^{-c\ep^{-\frac{1}{5}}}\right)
\end{equation}
by Lemma \ref{lzweiab}, (a). Combining
\[ \hat{T}_{q,j} - \chi_{j,\ep} P \chi_{j,\ep} = \left(\hat{T}_{q,j} - P\right) + \left( P -
\sum_{k=0}^m \chi_{k,\ep} P \chi_{k,\ep} \right) + \sum_{k\neq j} \chi_{k,\ep} P \chi_{k,\ep} \]
with \eqref{TundP}, \eqref{Pzerleg} and \eqref{zweiteab} proves (d).
\end{proof}

Since Theorem \ref{satz1} compares the eigenvalues of a self adjoint unbounded operator
on $\ell^2\left(\disk\right)$ with the eigenvalues of the
harmonic oscillator, which is an unbounded self adjoint operator on $L^2\left(\R^d\right)$,
we have to compare some
matrix elements with respect to the scalar product $\skpd{.}{.}$ with
those with respect to $\skpR{.}{.}$. How this can be done is shown in
the next lemma, giving an estimate for the difference of these terms.

\begin{Lem}\label{diskkont}
Let $T_{\ep,q,j}$ and $\hat{T}_{q,j}$ be defined in \eqref{tpiqj}
and \eqref{tqj} respectively and let $V_{\ep}^j$ be given by
\eqref{Vdj}.
Let $f,g\in L^2\left(\R^d\right)$ denote normalized
eigenfunctions of $\hat{H}^j$ given in \eqref{tildeHj} (of
the form \eqref{gnj}) and $f^\ep,g^\ep\in
\ell^2\left(\disk\right)$ their restrictions to the lattice. Let
$\chi_{j,\ep},\, 1\leq j \leq m$, $\tilde{\phi}_{0,\ep}$ and $\phi_{0,\ep}$ be
the cut-off functions defined in \eqref{defchi} and \eqref{phiaufR}.
Then for $\ep$ sufficiently small
\begin{enumerate}
\item for any $\alpha<\frac{1}{2}$
\begin{multline*}
\skpd{\chi_{j,\ep} f^{\ep}\,}{\Op_\ep^{\T}(\tilde{\phi}_{0,\ep})T_{\ep,q,j}
\Op_\ep^{\T}(\tilde{\phi}_{0,\ep})\,\chi_{j,\ep} g^{\ep}} \\
= \ep^{-d}\left(\skpR{\chi_{j,\ep} f\,}{\Op_\ep
(\phi_{0,\ep})\hat{T}_{q,j}\Op_\ep (\phi_{0,\ep})\,\chi_{j,\ep} g} +
O\left(\ep^{1 + \alpha}\right)\right)\, .
\end{multline*}
\item
\[\skpd{f^{\ep}\,}{\,\chi_{j,\ep} V_{\ep}^j\chi_{j,\ep}\,g^{\ep}} =
\ep^{-d}\left( \skpR{f\,}{\,\chi_{j,\ep} V_{\ep}^j\chi_{j,\ep}\,g} +
O\left(\ep^{\frac{13}{10}}\right)\right) \,.\]
\end{enumerate}
\end{Lem}

\begin{rem} The estimate in (b) is a rough Corollary of Lemma \ref{lzweiab}.
\end{rem}

\begin{proof}

(a):\\
Let $\tilde{t}_{\pi,q,j}$ and $t_{q,j}$ be defined in \eqref{tpiqj} and \eqref{tqj} respectively. Then we observe
that $\tilde{\phi}^2_{0,\ep}\tilde{t}_{\pi,q,j}$ on $\T^d$ can be identified with the function
$G:= \phi^2_{0,\ep} t_{q,j}$ on $\R^d$,
since $\supp G\in (-\pi,\pi)^d$ for $\ep$ sufficiently small. Setting
\begin{align}\label{u12v12}
u_1&:= {\mathscr F}_\ep^{-1}(\chi_{j,\ep}f^{\ep})\, , \qquad u_2:= F_\ep^{-1}(\chi_{j,\ep}f)\\
v_1&:= {\mathscr F}_\ep^{-1}(\chi_{j,\ep}g^{\ep})\, , \qquad v_2:= F_\ep^{-1}(\chi_{j,\ep}g)\, ,\nonumber
\end{align}
we obtain by use of \eqref{dnachT} that the left hand side of (a) is given by
\begin{equation}\label{zerlegung1}
 \skpR{u_1}{G v_1} = I_1 + I_2 + I_3\, ,
\end{equation}
where
\begin{equation}\label{I1undI2}
I_1 = \skpR{u_1 - u_2}{G v_1}\, , \qquad I_2 = \skpR{u_2}{G (v_1 - v_2)}
\end{equation}
and
\begin{equation}\label{I3}
I_3 = \skpR{u_2}{G v_2} = \ep^{-d} \skpR{\chi_{j,\ep}f}{ \Op_{\ep}(\phi_{0,\ep}) \hat{T}_{q,j}\Op_{\ep}(\phi_{0,\ep}) \chi_{j,\ep}g}\, ,
\end{equation}
where the last equalitiy follows from the ``Parseval'' relation \eqref{parsevalkont}
for the $\ep$-Fourier transform $F_\ep$ defined in \eqref{F}.
We claim that for any $\alpha<\frac{1}{2}$ and for $j=1,2$
\begin{equation}\label{Ij}
|I_j| = O(\ep^{-d + 1 + \alpha})\, ,
\end{equation}
which together with \eqref{I3} proves (a).

Using Cauchy-Schwarz in \eqref{I1undI2}, we obtain by Lemma \ref{lzweiab}, (b), and the boundedness of $G$ for any
$s<\frac{1}{2}$
\begin{equation}\label{I1CS}
|I_1| \leq \| u_1 - u_2 \|_{L^2} \| G v_1\|_{L^2} =  \| h_1 + h_2 \|_{L^2(|\xi|< \ep^s)}
\| G v_1\|_{L^2(|\xi|< \ep^s)} + O(\ep^\infty)\, ,
\end{equation}
where, setting $Q_x := [x,x+\ep]^d$,
\begin{align}\label{h1}
h_1(\xi) &= (\ep \sqrt{2\pi})^{-d} \sum_{x\in\disk}\int_{Q_x}\left((e^{\frac{i}{\ep} x\cdot \xi} -
e^{\frac{i}{\ep} y\cdot \xi}\right)\chi_{j,\ep} f(x)\, dy \\
h_2(\xi) &= (\ep \sqrt{2\pi})^{-d} \sum_{x\in\disk}\int_{Q_x} e^{\frac{i}{\ep} y\xi} h(y) \, dy\, ,
\end{align}
with
\[ h(y) = (\chi_{j,\ep} f)(y) - (\chi_{j,\ep}f)(x)\, , \qquad y\in Q_x \; . \]
Thus we have
$h_2 (\xi) = F_{\ep}^{-1} h (\xi)$, giving by \eqref{parsevalkont}
\begin{equation}\label{pars}
\| h_2\| _{L^2} = \| F_\ep^{-1}h\|_{L^2} = \ep^{-\frac{d}{2}} \|h\|_{L^2}\; .
\end{equation}
Using Lemma \ref{lzweiab}, (a), we obtain for any $s<\frac{1}{2}$ and $\ep$ sufficiently small
\begin{equation}\label{hab}
\|h\|_{L^2}^2 = \int_{|x-x_j|<\ep^s} |h(y)|^2\, dx + O(\ep^\infty)\, .
\end{equation}
In the domain of integration we have $\chi_{j,\ep} = 1$ for $s\in(\frac{2}{5}, \frac{1}{2})$ and $\ep$ sufficiently small.
This gives, applying the chain rule to the scaled function $f$,
\begin{equation}\label{hbetrag}
|h(y)| \leq \ep \sup_{z\in\R^d} |\nabla f(z)|\leq C \ep^{-\frac{d}{4}} \ep^{\frac{1}{2}} \; .
\end{equation}
Thus
\begin{equation}\label{hab2}
\left(\int_{|x-x_j|<\ep^s} |h(y)|^2\, dy\right)^{\frac{1}{2}} \leq  C \ep^{-d(\frac{1}{2} - s) + 1}
\end{equation}
Combining \eqref{hab2}, \eqref{hab} and \eqref{pars} gives, taking $\frac{1}{2}-s$ small,
\begin{equation}\label{h2ab2}
\|h_2 \|_{L^2} \leq C \ep^{-\frac{d}{2}} \ep^{\frac{d}{2}(\frac{1}{2}-s) + \frac{1}{2}} =
O\left(\ep^{-\frac{d}{2}}\ep^\alpha\right)\qquad\text{for any}\quad \alpha<\frac{1}{2}\, .
\end{equation}
To estimate $h_1(\xi)$, observe that for $y\in Q_x$
\[
\left|e^{\frac{i}{\ep}x \xi} - e^{\frac{i}{\ep} y \xi}\right|\,
\leq\, \sup_{ y\in Q_x} \left|\tfrac{1}{\ep}(x- y)\cdot\xi\right|
\,\leq\, C |\xi|
\]
uniformly in $x$ and $\xi$. Inserting this into \eqref{h1} and setting
\begin{equation}\label{tildef}
f(x) = \ep^{-\frac{d}{4}} \tilde{f}\left(\frac{x-x_j}{\sqrt{\ep}}\right)
\end{equation}
gives by \eqref{lemma2.10.2}
\begin{equation}\label{h1ab}
|h_1(\xi)| \leq C \sum_{x\in\disk} |\xi|\, |f(x)| \leq C|\xi| \ep^{-\frac{d}{4}}\sum_{y\in (\sqrt{\ep}\Z)^d}
\left|\tilde{f}\left(y-\frac{x_j}{\sqrt{\ep}}\right)\right|\leq C|\xi| \ep^{-\frac{3d}{4}}\, .
\end{equation}
From \eqref{h1ab}, we obtain
\[
\| h_1\|^2_{L^2(|\xi|<\ep^s)} \leq C \ep^{-\frac{3d}{2}} \int_{|\xi|<\ep^s} |\xi|^2\, d\xi \leq
C \ep^{-d}\ep^{d(\frac{1}{2}-s) + 2s}  \, . \]
Thus, taking $\frac{1}{2}-s$ small, we get for any $\alpha < \frac{1}{2}$
\begin{equation}\label{h1ab3}
\|h_1\|_{L^2(|\xi|<\ep^s)} \leq C_\alpha \ep^{-\frac{d}{2} + \alpha}\; .
\end{equation}
Furthermore, since $\sup_{|\xi|<\ep^s} |t_{q,j}(\xi)| \leq C \ep^{2s}$, we get using \eqref{u12v12} and \eqref{dnachT}
\begin{equation}\label{v1ab}
\|G v_1\|_{L^2(|\xi|<\ep^s)} \leq C \ep^{2s} \|\chi_{j,\ep}g^\ep\|_{\ell^2} = O\left( \ep^{-\frac{d}{2}+2s}\right)\; .
\end{equation}
Combining \eqref{v1ab}, \eqref{h1ab3}, \eqref{h2ab2} and \eqref{I1CS} proves \eqref{Ij} for $I_1$.
The estimate for $I_2$ is similar.

(b):\\
Using the identity $\ep^d = \int_{Q_x}\,d y$ and setting $W:=\chi_{j,\ep} V_{\ep}^j\chi_{j,\ep}$, the left hand side of (b) can
analog to \eqref{zerlegung1} be written as
\begin{equation}\label{3terme2}
 \skpd{f^{\ep}}{W\,g^{\ep}} = I_1 + I_2 + I_3\, ,
\end{equation}
where
\begin{align}\label{IeinsIzwei}
I_1 &= \frac{1}{\ep^d}\sum_{x\in\disk}\int_{Q_x}  (f(x) - f( y))\,W g(x)\, dy \\
I_2 &= \frac{1}{\ep^d}\sum_{x\in\disk}\int_{Q_x}
 f( y)\,(W g(x) - W g(y))\, dy\nonumber
\end{align}
and
\begin{equation}
I_3 = \ep^{-d} \int_{\R^d} f(y)\,W g(y) \, dy = \ep^{-d} \skpR{f\,}{\,\chi_{j,\ep} V_{\ep}^j\chi_{j,\ep}\,g}\: .
\end{equation}
Then, similar to the proof of (a), it remains to estimate $I_1$ and $I_2$. We claim that
\begin{equation}\label{I_j2}
|I_j| = O\left(\ep^{\frac{13}{10}-d}\right)\, , \qquad j=1,2 \; .
\end{equation}
By the scaling of $\chi_{j,\ep}$
\begin{equation}\label{Wab}
\sup_{x\in\R^d} |W(x)|\leq \sup_{|x|\leq 2\ep^{\frac{2}{5}}} |V_\ep^j(x)| = O\left(\ep^{\frac{4}{5}}\right)\, ,
\end{equation}
since $V_\ep^j$ is quadratic in $x$.
Combining \eqref{vterm1e} and \eqref{Wab} shows \eqref{I_j2} for $j=1$. The proof for $I_2$ is similar.
\end{proof}

We still need one more estimate for the proof of Theorem
\ref{satz1}. It concerns replacing the
$x$-dependent quadratic approximation $\hat{T}_q$ of the kinetic energy
by the operator $\hat{T}_{q,j}$ fixed at the well $x_j$.

\begin{Lem}\label{TqTqj}
Let $\hat{T}_q$ and $\hat{T}_{q,j}$ be given by \eqref{tq1} and \eqref{tqj} respectively
for $1\leq j \leq m$.
Let $\chi_{j,\ep}$ be the cut-off function defined in \eqref{defchi} and
$f,g$ denote normalized eigenfunctions of $\hat{H}^j$ given in \eqref{tildeHj},
then
\[ \left| \skpR{f}{\chi_{j,\ep}\,\hat{T}_q\,\chi_{j,\ep}\, g} -\skpR{f}{\chi_{j,\ep}\,\hat{T}_{q,j}\,\chi_{j,\ep}\, g}\right| =
O\left(\ep^{\frac{7}{5}}\right) \, . \]
\end{Lem}

\begin{proof}
By the definition of the operators $\hat{T}_q$ and $\hat{T}_{q,j}$
\begin{multline*}
\left|\skpR{f}{\chi_{j,\ep}\,(\hat{T}_q-\hat{T}_{q,j})\,\chi_{j,\ep}\,
g}\right| \\
=  \left| \skpR{f}{\chi_{j,\ep}\,\left[\ep^2\sum_{\nu ,
\mu}(B_{\nu\mu}(x)-B_{\nu\mu}(x_j))\partial_\nu
   \partial_\mu + \ep (t_1(x,0)-t_1(x_j,0))\right]\,\chi_{j,\ep}\, g}\right| \, .
\end{multline*}
As $g$ is scaled by $\ep^{-\frac{1}{2}}$,
\begin{equation}\label{381}
\| \ep^2\partial_\nu\partial_\mu \chi_{j,\ep} g \|_{L^2} = O(\ep) \, .
\end{equation}
Since $|x-x_j|\leq 2\ep^{\frac{2}{5}}$ in the support of $\chi_{j,\ep}$, we
have by Hypothesis \ref{hypsatz},(a),(i) that $B_{\nu\mu}(x) - B_{\nu\mu}(x_j) =
O\left(\ep^{\frac{2}{5}}\right)$ and $t_1(x,0)-t_1(x_j,0) = O(\ep^{\frac{2}{5}})$. Together with \eqref{381},
this estimate proves the lemma by use of the Schwarz inequality.
\end{proof}

\section{Proof of Theorem \ref{satz1}}\label{sec223}

Following Simon \cite{Si1}, we prove equality in \eqref{limen} by proving an upper and a lower estimate. For
the sake of the reader, we give a complete proof.

\subsection{Estimate from above}
\begin{equation}\label{i}
  \frac{E_n(\ep)}{\ep}\, \leq \,e_n + O\left(\ep^{\frac{1}{5}}\right) \qquad
  \text{as}\quad \ep\to 0\; .
\end{equation}

At first we use the points (a) and (c) of Lemma \ref{chiphi},
leading to the estimate
\begin{eqnarray}\label{absch3}
\lefteqn{\skpd{g^{\ep}_{\alpha j}}{H_{\ep}\,g^{\ep}_{\beta l}}
 = \skpd{g^{\ep}_{\alpha j}}{\chi_{j,\ep} H_{\ep}\chi_{j,\ep}\,g^{\ep}_{\beta l}} +
 O\left(\ep^{\frac{6}{5}}\right)
  }\hspace{0.8cm}\\
&& = \skpd{g^{\ep}_{\alpha j}}{\chi_{j,\ep}\Op_\ep^{\T}(\tilde{\phi}_{0,\ep})T_\ep\Op_\ep^{\T}(\tilde{\phi}_{0,\ep})
\chi_{j,\ep}\,g^\ep_{\beta l})} +
\skpd{g^{\ep}_{\alpha j}}{\chi_{j,\ep} V_{\ep}\chi_{j,\ep}\,g^{\ep}_{\beta l}}
+
    O\left(\ep^{\frac{6}{5}}\right) \, .\nonumber
\end{eqnarray}
By Proposition \ref{micnorm} and by \eqref{chiabV} (quadratic approximation of $T_\ep$ localized at $\xi=0, x=x_j$ and of
$V_\ep$ localized at $x=x_j$) we have
\begin{eqnarray}\label{absch2}
\lefteqn{ \skpd{g^{\ep}_{\alpha j}}{H_{\ep}\,g^{\ep}_{\beta l}} }\hspace{0.3cm}\\
&&=\skpd{g^\ep_{\alpha j}}{\chi_{j,\ep} \Op_\ep^{\T}(\tilde{\phi}_{0,\ep})T_{\ep,q,j}\Op_\ep^{\T}(\tilde{\phi}_{0,\ep})
\,\chi_{j,\ep} g^\ep_{\beta l}} +
\skpd{g^{\ep}_{\alpha j}}{\chi_{j,\ep} V_{\ep}^j\chi_{j,\ep}\,g^{\ep}_{\beta l}} +
   O\left(\ep^{\frac{6}{5}}\right)  \nonumber\\
&& = \ep^{-d}\left( \skpR{ g_{\alpha j}}{\chi_{j,\ep}\Op_\ep(\phi_{0,\ep})\hat{T}_{q,j}
\Op_\ep (\phi_{0,\ep}) \chi_{j,\ep}g_{\beta l}} + \skpR{g_{\alpha j}}{\chi_{j,\ep}
V_{\ep}^j\chi_{j,\ep}\,g_{\beta l}} +
   O\left(\ep^{\frac{6}{5}}\right)\right) \, ,\nonumber
\end{eqnarray}
where for the second step, the transition from (functions and
scalar product in) $\ell^2\left(\disk\right)$ to $L^2\left(\R^d\right)$,
we used Lemma \ref{diskkont},(a) and (b).
Point (b) and (d) of Lemma \ref{chiphi} and \eqref{tildeHj} yield
\begin{equation}\label{absch1}
\ep^d(\text{rhs}\eqref{absch2})
 = \skpR{g_{\alpha j}}{\hat{H}^{j}\,g_{\beta l}}
 +
   O\left(\ep^{\frac{6}{5}}\right) = \skpR{g'_{\alpha j}}{H^{j}\,g'_{\beta l}}
 +
   O\left(\ep^{\frac{6}{5}}\right)\: ,
\end{equation}
where the second equality follows from the fact that $H^j$ and $\hat{H}^j$ are unitarily equivalent
(see \eqref{Hjj}).
Since by definition $H^j g'_{\alpha j} = \ep e_{n(\alpha, j)} g'_{\alpha j}$, the estimates
\eqref{absch1} and \eqref{absch2} can be combined to give
\begin{equation}\label{erwartung}
\skpd{g^{\ep}_{\alpha j}}{H_{\ep}\,g^{\ep}_{\beta l}}
=\ep^{-d}\left( \ep\,e_{n(\alpha, j)}\delta_{n(\alpha,j),n(\beta,l)} +
O\left(\ep^{\frac{6}{5}}\right)\right)\: ,
\end{equation}
where $n(\alpha,j)$ denotes the number of the eigenvalue corresponding to the pair
$(\alpha, j)$. We shall show that \eqref{erwartung} leads to \eqref{i} by use of the
Min-Max-principle. Choose $\zeta_1,\ldots \zeta_{n-1}$ in the domain of
$H_{\ep}$ and define
\begin{equation}\label{Qdef}
Q(\zeta_1,\ldots \zeta_{n-1})
:=\inf\left\{\skpd{\psi}{H_{\ep}\,\psi}\:|\:\psi\in
\De (H_\ep),\: \|\psi\|=1,\:\psi\in [\zeta_1,\ldots
\zeta_{n-1}]^{\perp}\right\}
\end{equation}
and
\begin{equation}\label{mu}
\mu_n(\ep) :=
\sup_{\zeta_1,\ldots\zeta_{n-1}}Q(\zeta_1,\ldots\zeta_{n-1})\: .
\end{equation}

For $\lambda >0$ we can choose $\zeta_1,\ldots\zeta_{n-1}$, such
that
\begin{equation}\label{Emu}
\mu_n(\ep) \leq Q(\zeta_1,\ldots\zeta_{n-1}) + \lambda\: .
\end{equation}
It follows from Lemma \ref{kleinnorm} that for $\ep>0$ sufficiently small the functions
$g^\ep_{n(\alpha,j)}:= \ep^{\frac{d}{2}} g^\ep_{\alpha j}$ satisfy
\begin{equation}\label{almor}
\skpd{g^\ep_n}{g^\ep_m} = \delta_{n,m} + O(\sqrt{\ep})\, ,
\end{equation}
in particular they are linearly independent
and $\mathscr{M}_n:= \Span \{g^\ep_m\,|\, m \leq n\}$ has dimension $n$.
Then ${\mathscr N} := {\mathscr M}_n \cap
[\zeta_1,\ldots \zeta_{n-1}]^{\perp}$ is at least one dimensional.
Thus there exists a function $\psi\in{\mathscr N}$ with $\|\psi\|_{\ell^2} = 1$ and it
follows from \eqref{erwartung}, \eqref{Qdef} and \eqref{almor}
that
\begin{equation}\label{qab}
Q(\zeta_1,\ldots \zeta_{n-1}) \leq \skpd{\psi}{H_{\ep}\,\psi}
\leq\, \ep\,e_n + O\left(\ep^{\frac{6}{5}}\right)
\end{equation}
Since $\lambda$ is arbitrary, we have by  (\ref{Emu}) and
(\ref{qab})
\begin{equation}\label{mun}
\mu_n(\ep) \leq \ep\, e_n + O\left(\ep^{\frac{6}{5}}\right)\, .
\end{equation}

By Theorem \ref{perssontheo},
Hypothesis \ref{hypsatz} ensures that $\inf\sigma_{ess} (H_\ep)\geq c >0$ uniformly in $\ep$ for $\ep$ sufficiently small.
Since $\mu_n(\ep)$ is by \eqref{mun} of order $\ep$, for $\ep$ small enough it follows from the
Min-Max-principle that $\mu_n(\ep)$ belongs to the discrete
spectrum and coincides with $E_n(\ep)$.

\subsection{Estimate from below}
\begin{equation}\label{ii}
  \frac{E_n(\ep)}{\ep}\, \geq \,e_n + O\left(\ep^{\frac{1}{5}}\right)\qquad\text{as}
  \quad \ep \to 0\;
\end{equation}

For $n>1$, let $l\leq n-1$ be such that $e_n = e_{n-1}=\ldots = e_{l+1} > e_l$
and set $e\in(e_l,e_n)$, for $n=1$ choose $e < e_1$ (in particular
$e\not\in\sigma(\bigoplus_j K_j)$). Then we claim that there exists a constant $C>0$ such that
\begin{equation}\label{ungl1}
\skpd{\psi}{H_{\ep}\,\psi} \geq \ep\, e\skpd{\psi}{\psi} +
\skpd{\psi}{R_l\,\psi} - C \ep^{\frac{6}{5}} \|\psi\|^2_{\ell^2} \: ,\qquad \psi\in
\De(H_\ep)\: ,
\end{equation}
for some symmetric operator $R_l$ with $\rank R_l \leq l$. This
implies \eqref{ii}. To see this implication, let $\psi\in\mathscr{E}_n:= \Span\{ h_k\in\ell^2(\disk)\, |\,
h_k\;\text{is the}\;k\,\text{th eigenfunction of}\; H_\ep, \|h_k\|_{\ell^2} = 1, k\leq n\}$. From the Min-Max-formula it
follows that
\begin{equation}\label{Er+1}
E_{n}(\ep) \geq \skpd{\psi}{H_{\ep}\,\psi}\: .
\end{equation}
On the other hand there exists a $\psi\in{\mathscr E}_n\cap \ker R_l$,
since $\dim\ker \left(R_l|_{{\mathscr E}_n}\right) \geq 1$.
For this $\psi$ the inequality \eqref{ungl1} yields
\begin{equation}\label{ungl2}
\skpd{\psi}{H_{\ep}\,\psi} \geq \ep\, e + O\left(\ep^{\frac{6}{5}}\right)\: ,
\end{equation}
which together with \eqref{Er+1} gives \eqref{ii}.
It therefore suffices to show \eqref{ungl1}.

By Lemma \ref{1zer}, $H_\ep$ splits as
\begin{equation}
H_{\ep} = \sum_{j=1}^m\chi_{j,\ep}H_{\ep}\chi_{j,\ep} +
\chi_{0,\ep}H_{\ep}\chi_{0,\ep} +
O\left(\ep^{\frac{6}{5}}\right)\: ,
\end{equation}
where the estimate on the error term in the following estimates is
understood with respect to operator norm. $\chi_{0,\ep}$ is
supported in the region outside of the wells, thus
$|x-x_j|>\ep^{\frac{2}{5}}$ for $1\leq j \leq m$ and $x\in \supp
\chi_{0,\ep}$. Since the kinetic term is positive modulo terms of order $\ep$ and the potential is of second order in $x$ or of
order $\ep$, we have for $\ep$ sufficiently small, $e<e_n$ and some constant $\tilde{C}>0$
\begin{equation}\label{chi0}
\chi_{0,\ep}H_{\ep}\chi_{0,\ep} \geq
\chi_{0,\ep}V_{\ep}\chi_{0,\ep} \geq ( -C\ep + \tilde{C}\ep^{\frac{4}{5}})\chi_{0,\ep}^2  \geq \ep\,e\,\chi_{0,\ep}^2 \: .
\end{equation}
In the neighborhoods of the wells, (\ref{chiabV}) allows to
approximate the potential by the quadratic term, therefore
\eqref{chiabV} and \eqref{chi0} give
\begin{equation}\label{phi1plusphi0}
H_{\ep} \geq \sum_{j=1}^m\chi_{j,\ep}(T_{\ep} +
V_{\ep}^j)\chi_{j,\ep} + \ep\, e\,\chi_{0,\ep}^2 +
O\left(\ep^{\frac{6}{5}}\right)\: .
\end{equation}
In the first summand we introduce the partition of unity $\tilde{\phi}_{k,\ep}, k=1,0$ in
momentum space, defined in \eqref{phiaufR}, and get by Lemma \ref{1zer}
\begin{multline}\label{tplusv}
\sum_{j=1}^m\chi_{j,\ep}(T_{\ep} + V_{\ep}^j)\chi_{j,\ep} =
\sum_{j=1}^m\chi_{j,\ep}(x)\Op_\ep^{\T}(\tilde{\phi}_{0,\ep})
(T_{\ep} + V^j_{\ep})\Op_\ep^{\T}(\tilde{\phi}_{0,\ep})\chi_{j,\ep}(x) +\\
 + \sum_{j=1}^m\chi_{j,\ep}(x)\Op_\ep^{\T}(\tilde{\phi}_{1,\ep})(T_{\ep} + V^j_{\ep})
\Op_\ep^{\T}(\tilde{\phi}_{1,\ep})\chi_{j,\ep}(x)
  +  O\left(\ep^{\frac{6}{5}}\right)\: .
\end{multline}
By Proposition \ref{micnorm}, modulo terms of order $O\left(\ep^{\frac{6}{5}}\right)$, it is  possible to replace
$T_{\ep}$ by $T_{\ep,q,j}$ near $\xi=0$ and $x=x_j,\, j=1,\ldots m$. The
function $\tilde{\phi}_{1,\ep}$ is supported in the exterior region with $|\xi|>\ep^{\frac{2}{5}}$, thus we
have by arguments similar to those leading to (\ref{chi0})
\begin{equation}\label{phiaussen}
\Op_\ep^{\T}(\tilde{\phi}_{1,\ep})(T_{\ep} + V_{\ep}^j)\Op_\ep^{\T}(\tilde{\phi}_{1,\ep}) \geq \ep\,
e\,\Op_\ep^{\T}(\tilde{\phi}_{1,\ep})^2 \: .
\end{equation}
Substituting (\ref{phiaussen}) in (\ref{tplusv}), replacing $T_\ep$ by $T_{\ep,q,j}$ in the first
summand of \eqref{tplusv} and substituting the resulting equation in (\ref{phi1plusphi0})
yields
\begin{align}\label{Hd1}
 H_{\ep} &\geq M
 + \ep\,
e\sum_{j=1}^m\chi_{j,\ep}(x)\left(\Op_\ep^{\T}(\tilde{\phi}_{1,\ep})\right)^2\chi_{j,\ep}(x)
+
 \ep\, e\chi_{0,\ep}^2 + O\left(\ep^{\frac{6}{5}}\right)\, ,\quad\text{where}\\
M &:= \sum_{j=1}^m\chi_{j,\ep}(x)\Op_\ep^{\T}(\tilde{\phi}_{0,\ep})\,
(T_{\ep,q,j} +
V_{\ep}^j)\,\Op_\ep^{\T}(\tilde{\phi}_{0,\ep})\chi_{j,\ep}(x)\nonumber
\end{align}
By the isometry of the Fourier transform
\begin{align}\label{Hd2}
\skpd{\psi}{M\psi} &= \sum_{j=1}^m\skpT{\tilde{\phi}_{0,\ep} {\mathscr F}_\ep^{-1}(\chi_{j,\ep}\psi)}{(t_{\pi,q,j} + {\mathscr F}_\ep^{-
1}V_{\ep}^j{\mathscr F}_\ep)\, \tilde{\phi}_{0,\ep} {\mathscr F}_\ep^{-1}(\chi_{j,\ep}\psi)} \\
&= \sum_{j=1}^m \skpR{\phi_{0,\ep}{\mathscr
F}_\ep^{-1}(\chi_{j,\ep}\psi)}{(F_\ep^{-1}H^jF_\ep)\,\phi_{0,\ep}{\mathscr
F}_\ep^{-1}(\chi_{j,\ep}\psi)}\nonumber\; .
\end{align}
In the last step we used that for $\ep$ sufficiently small we
can replace the scalar product in $
L^2\left(\T^d\right)$ by the scalar
product in $L^2\left({\R}^d\right)$, if we
simultaneously replace $\tilde{\phi}_{0,\ep}$ by $\phi_{0,\ep}$ and
$t_{\pi,q,j}$ by $t_{q,j}$. This follows from the fact that the
range of the integral is in both cases restricted to the support
of $\phi_{0,\ep}$. Moreover changing variables
allows to replace ${\hat H}^j = \widehat{T}_{q,j} + V_\ep^j$ by $H^j$ (see \eqref{Hjj}) and
${\mathscr F}_\ep^{-1}V_{\ep}^j{\mathscr F}_\ep =
F_\ep^{-1}V_{\ep}^jF_\ep$ and $F_\ep\skp{\xi}{\xi}F_\ep^{-1} = -
\ep^2\Delta$.

We introduce the spectral decomposition of $F_\ep^{- 1}H^jF_\ep$.
Denote by $e_{k,j}$ the $k$th eigenvalue of $H^j$ and by $l_j$ the number of eigenvalues of $H^j$ not exceeding $e$.
Thus $e_{l_j}\leq e_l < e$ for all $j$ and $\sum_{j=1}^m l_j = l$.
By replacing all eigenvalues $e_{k_j}>e$ of $H^j$ by $e$ we get
 \begin{equation}\label{Hd3}
(F_\ep^{-1}H^jF_\ep) = \ep \sum_{k} e_{k,j} \Pi^j_{k} \geq \ep \,\sum_{k\leq l_j}e_{k,j}\,\Pi^j_{k} +
\ep\, e\,(\id - \sum_{k\leq l_j}\Pi^j_{k})\: ,
\end{equation}
where $\Pi^j_{k} $ denotes the projection on the eigenspace of $e_{k,j}$. Inserting \eqref{Hd3} into the
right hand side of \eqref{Hd2} and replacing
$\phi_{0,\ep}$ by $\tilde{\phi}_{0,\ep}$ and $\skpR{~.~}{~.~}$ by $\skpT{~.~}{~.~}$ yields
\begin{multline}\label{Hd6}
\skpd{\psi}{M\psi}
\geq \sum_{j=1}^m \left\{ \skpT{\tilde{\phi}_{0,\ep} {\mathscr
F}_\ep^{-1}(\chi_{j,\ep}\psi)}{\ep\, \sum_{k}(e_{k,j}-e)
\Pi^j_{k}\tilde{\phi}_{0,\ep}{\mathscr F_\ep}^{-1}(\chi_{j,\ep}\psi)}\right.\\
\left. +
\ep\,e\, \skpT{\tilde{\phi}_{0,\ep}{\mathscr F}_\ep^{-1}(\chi_{j,\ep}\psi)}{\tilde{\phi}_{0,\ep}{\mathscr
F}_\ep^{-1}(\chi_{j,\ep}\psi)}\right\}
\end{multline}
Thus by (\ref{Hd6}) together with (\ref{Hd1})
there exists a constant $C>0$ such that
\begin{multline}\label{Hd4}
 \skpd{\psi}{H_{\ep}\,\psi} \geq \ep\,
e\,\sum_{j=1}^m\skpT{\tilde{\phi}_{0,\ep}{\mathscr F}_\ep^{-1}(\chi_{j,\ep}\psi)}
{\tilde{\phi}_{0,\ep}{\mathscr F_\ep}^{-1}(\chi_{j,\ep}\psi)} + \skpT{{\mathscr F}_\ep^{-1}\psi}{R_l{\mathscr F_\ep}^{-1}\psi)}\\
 +\ep\, e\,\sum_{j=1}^m\skpT{\tilde{\phi}_{1,\ep}{\mathscr
F}_\ep^{-1}(\chi_{j,\ep}\psi)}{\tilde{\phi}_{1,\ep}{\mathscr
F}_\ep^{-1}(\chi_{j,\ep}\psi)}\\
 + \ep\,e\,\skpT{({\mathscr F}_\ep^{-1}\chi_{0,\ep}\psi)}{(\tilde{\phi}^2_0 +
\tilde{\phi}^2_1)({\mathscr F_\ep}^{-1}\chi_{0,\ep}\psi)}
 - C \ep^{\frac{6}{5}} \|\psi\|^2_{\ell^2}\: ,
\end{multline}
where
 \begin{equation}\label{Rell}
R_l := \sum_{j=1}^m({\mathscr F_\ep}^{-1}\chi_{j,\ep}{\mathscr F_\ep})\tilde{\phi}_{0,\ep} A_j \tilde{\phi}_{0,\ep}({\mathscr F}_\ep^{-1}\chi_{j,\ep}{\mathscr F_\ep}) \, , \quad\quad
A_j := \sum_{k\leq
l_j}(\ep(e_{k,j} - e)\Pi^j_k) \; .
\end{equation}
Since $\rank (A+B)\leq\rank A + \rank B$ and $\rank \Pi^j_k =1$, the operator
$A_j$ has rank at most $l_j$.
Conjugation does not increase
the rank and moreover
$\sum_{j=1}^m l_j = l$, thus we get $\rank R_l \leq l$.
Introducing $\tilde{\phi}^2_0 +\tilde{\phi}^2_1= \id$ in the fourth summand, we combine
this term with the first and
third summand and rewrite the rhs of \eqref{Hd4} as
\begin{multline}\label{Hd5}
\ep \, e \sum_{j=0}^m\skpT{\tilde{\phi}_{0,\ep}({\mathscr F}_\ep^{-1}\chi_{j,\ep}\psi)}
{\tilde{\phi}_{0,\ep}({\mathscr F}_\ep^{-1}\chi_{j,\ep}\psi)} +
\skpT{{\mathscr F}_\ep^{-1}\psi}{R_l\,({\mathscr F_\ep}^{-
1}\psi)} \\
+ \ep\, e \sum_{j=0}^m\skpT{\tilde{\phi}_{1,\ep}({\mathscr F}_\ep^{-1}\chi_{j,\ep}\psi)}
{\tilde{\phi}_{1,\ep}({\mathscr F_\ep}^{-
1}\chi_{j,\ep}\psi)} - C \ep^{\frac{6}{5}}\|\psi\|^2_{\ell^2}\: .
\end{multline}
Again the first and third summand can be combined so that the
cut-off functions in both spaces add up to $\id$. We thus get by
(\ref{Hd4}) and (\ref{Hd5}), for some $C>0$,
\begin{equation}
\skpd{\psi}{H_{\ep}\,\psi} \geq  \ep\, e\skpd{\psi}{\psi} +
\skpd{\psi}{B_l\psi} - C \ep^{\frac{6}{5}}\|\psi\|^2_{\ell^2}\: ,
\end{equation}
where $B_l:={\mathscr F_\ep}R_l{\mathscr F_\ep}^{-1}$ is again an operator of rank at most $l$.
Thus \eqref{ungl1} holds.
Combined with \eqref{i}, this
completes the proof of Theorem \ref{satz1}.

\begin{appendix}

\section{Pseudo-differential operators in the discrete setting}\label{symbols}

In the following, some properties of the symbols given in Definition
\ref{symbolspaces} and of the associated operators are collected.
For the sake of the reader, we recall the definitions of the $h$-scaled symbol classes $S_\delta^k(m)(\R^d)$ and  of the associated pseudo-differential operators
(see Dimassi-Sj\"ostrand \cite{dima} and Robert \cite{robert}).

\begin{Def}\label{defpsdokont}
\ben
\item
A function $m: \R^d\rightarrow [0,\infty)$ is called an order function,
if there exist constants $C_0>0$ and
$N_0>0$ such that
\[ m(x) \leq C_0\langle x-y \rangle^{N_0} m(y) \, , \qquad x,y\in\R^d\, .\]
\item
For $\delta\in[0,1]$, the space
$S^k_\delta(m)\left(\R^d\right)$ consists of
functions $a(x;\ep)$ on $\R^d\times (0,1]$, such that there exists constants $C_\alpha >0$ such that for all
$x\in\R^d, \ep\in (0,1]$
\[   |\partial^{\alpha}_x a(x;\ep)| \leq C_{\alpha}
m(x)\ep^{k - \delta |\alpha|} \, . \]
\item
Let $a_j\in S_\delta^{k_j}(m)(\R^d), k_j\nearrow \infty$, then $a\sim
\sum_{j=0}^\infty a_j$ means that
$a - \sum_{j=0}^Na_j \in S_\delta^{k_{N+1}}(m)(\R^d)$ for every $N\in \N$.
\item A pseudo-differential operator $\Op_\ep: \Ce_0^\infty
\left(\R^d\right) \rightarrow \left(\Ce_0^\infty\right)'
\left(\R^d\right)$ associated to a symbol $a\in
S^k_\delta(m)\left(\R^{2d}\right)$ is defined by
\begin{equation}\label{opkont}
\Op_\ep(a) u(x) = \frac{1}{(\ep 2\pi)^d} \int\limits_{\R^{2d}}
e^{\frac{i}{\ep}(y-x)\xi} a(x,\xi;\ep) u(y) \,dy d\xi\; ,
\qquad u\in \Ce_0^\infty \left(\R^d\right) \; .
\end{equation}
\een
\end{Def}

We start this section showing that by use of the identification of functions on the torus $\T^d$ with $2\pi$-periodic functions
on $\R^d$, the discrete operator $\Op_\ep^{\T}(a)$ associated to the symbol
$a\in S_\delta^r(m)\left(\R^d\times \T^d\right)$
can be understood as a special case of the operator $\Op_\ep (b)$ associated to the symbol
$b\in S_\delta^r(m)\left(\R^{2d}\right)$.

First we notice (see for example H\"ormander \cite{hormander2}) that for any $2\pi$-periodic function $g\in \Ce^\infty\left(\R^d\right)$ the Fourier transform $F_\ep g$ defined in \eqref{F} satisfies
\begin{equation}\label{opall1}
F_\ep g = \left(\frac{\ep}{\sqrt{2\pi}}\right)^d \sum_{z\in\disk} \delta_z c_z\, , \quad\text{where}\quad
c_z:= \int_{[-\pi,\pi]^d} e^{-\frac{i}{\ep}z\mu} g(\mu)\, d\mu \; .
\end{equation}
Thus for $a\in S_\delta^r(m)\left(\R^{2d}\right)$ with $a(x,\xi + 2\pi\eta) = a(x,\xi)$ for any $\eta\in\Z^d, x, \xi \in\R^d$ and
$u\in \mathscr{S}\left(\R^d\right)$ by \eqref{opkont}
\begin{align}
 \Op_\ep(a) u(x) &= \frac{1}{(\ep \sqrt{2\pi})^d}\left[\left(F_{\ep, \xi\to x} a\right)(x,\,.\,) \, * u\right](x) \nonumber\\
&= \frac{1}{(2\pi)^d} \sum_{y\in{\mathscr G}_x} \int_{[-\pi,\pi]^d} e^{\frac{i}{\ep}(x-y)\xi}
a(x,\xi;\ep) u(y)\, d\xi\; ,\label{Opasum}
\end{align}
where, as in Remark \ref{Gxnull}, ${\mathscr G}_x = \disk + x$.
If we denote by $r:\mathscr{S}\left(\R^d\right)\ra s\left(\disk\right)$ the restriction to the lattice $\disk$, \eqref{Opasum} implies
\begin{equation}\label{Opasum2}
r\circ \Op_\ep(a) u = \Op_\ep^{\T}(a) ru\, , \quad u\in\mathscr{S}\left(\R^d\right)\; .
\end{equation}

\begin{Lem}\label{opastetig}
Let $a\in S_\delta^0(m)\left(\R^d\times \T^d\right)$, then, for fixed $\ep>0$,  $\Op_\ep^{T}(a)$ defined in \eqref{psdo2} is continuous $:
s\left(\disk\right)\longrightarrow
s\left(\disk\right)$, where the space $s\left(\disk\right)$ with its natural Fr\'echet topology is defined in \eqref{diskschwarz}.
\end{Lem}

\begin{proof}
We will deduce the continuity of $\Op_\ep^{\T}(a)$ on $s\left(\disk\right)$ from the continuity of
$\Op_\ep$ on $\mathscr{S}\left(\R^d\right)$, which is proven e.g. in Grigis-Sj\"ostrand \cite{grisjo}.
To this end, we consider
a cut-off function $\zeta\in\Ce_0^\infty \left(\R^d\right)$ such that $\zeta(0)=1$ and $\zeta(x)=0$
for $|x|\geq \frac{1}{2}$.
We set $\zeta_{\ep,z}:= \zeta\left(\frac{1}{\ep}(x-z)\right)$ and define
\[ j: s(\disk) \ra \mathscr{S}\left(\disk\right)\, , \quad u\mapsto ju:= \sum_{z\in\disk} u(z) \zeta_{\ep,z} \; . \]
Then $r\circ j=\id_s$ and $\Op_\ep^{\T}(a) = r\circ \Op_\ep (a) \circ j$ by \eqref{Opasum2}. It remains to show that
$r$ and $j$ are continuous, which is straight forward with
$\|rf\|_{\alpha}\leq \|f\|_{\alpha,0}$ and $\|jv\|_{\alpha,\beta}\leq C_\beta\ep^{-|\beta|} \|v\|_\alpha$
for some $C_\beta>0$.
\end{proof}

We define for $u\in\mathscr{S}(\R^{2d})$
\begin{equation}\label{opall3}
e^{i\ep D_xD_\xi} u(x,\xi) := \frac{1}{(\ep 2\pi)^d} \int_{\R^{2d}} e^{-\frac{i}{\ep} z\eta} 
u(x-z, \xi-\eta)\, dz\,d\eta \; .
\end{equation}

The following lemma is an adapted and more detailed version of Dimassi-Sj\"ostrand \cite{dima}, Proposition 7.6.

\begin{Lem}\label{A(D)}
Let $0\leq \delta \leq \frac{1}{2}$ and $m$ be an order function.
Then $e^{i\ep D_xD_\xi} : \mathscr{S}'(\R^d\times\T^d)\ra \mathscr{S}'(\R^d\times \T^d)$ is continuous
$: S^r_\delta(m)\left(\R^d\times\T^d\right)\rightarrow
S^r_\delta(m)\left(\R^d\times\T^d\right)$.
If $\delta < \frac{1}{2}$, then
\begin{equation}\label{eDDent}
 e^{i\ep D_xD_\xi}b(x,\xi) \sim \sum_{j=0}^\infty \frac{1}{j!}\,\left((i\ep D_x\cdot D_\xi)^j\,b\right) (x,\xi)
\end{equation}
in $S_\delta^r(m)\left(\R^d\times\T^d\right)$. If we write
$e^{i\ep D_xD_\xi}b = \sum_{j=0}^{N-1} \frac{(i\ep D_x\cdot D_\xi)^j}{j!}\,b + R_N(b)$, the remainder $R_N(b)$
is an element of the symbol class $S_\delta^{r + N(1-2\delta)}(m)$ and the Fr\'echet-seminorms of $R_N$
depend (linearly) only on finitely many $\|b\|_{\alpha,\beta}$ with $|\alpha|,|\beta|\geq N$:
\begin{equation}\label{RNnorm}
\| R_N(b) \|_{\alpha,\beta} \leq \sum_{\natop{\alpha',\beta'\in\N^d}
{N\leq |\alpha'|,|\beta'|\leq M}} C_{\alpha,\beta,\alpha',\beta'} \| b\|_{\alpha',\beta'}
\end{equation}
for some $M\in N$ and $C_{\alpha,\beta,\alpha',\beta'}>0$ independent of $\ep\in(0,1]$ .
\end{Lem}

\begin{proof}
Since $S_\delta^r(m)\left(\R^d\times \T^d\right)$ injects continuously into $S_\delta^r(m)\left(\R^{2d}\right)$,
by \cite{dima}, Prop.7.6, $e^{i\ep D_xD_\xi}$ maps $S_\delta^r(m)\left(\R^d\times \T^d\right)$
continuously into $S_\delta^r(m)\left(\R^{2d}\right)$. Thus, to prove continuity, it remains
to show that $e^{i\ep D_xD_\xi}a$ is periodic with respect to $\xi$ for $a\in S_\delta^r(m)\left(\R^d\times
\T^d\right)$.
Since $e^{i\ep D_xD_\xi} : {\mathscr S}'(\R^{2d}) \rightarrow {\mathscr S}'(\R^{2d}) $ is defined by
$\left(e^{i\ep D_xD_\xi}u\right)(\phi) := u\left(e^{i\ep D_xD_\xi}\phi\right)$, it suffices to prove
on ${\mathscr S}(\R^{2d})$
\begin{equation}\label{Ltaukomm}
e^{i\ep D_xD_\xi}\tau_\gamma = \tau_\gamma e^{i\ep D_xD_\xi}\, , \qquad \gamma\in \left(2\pi \Z\right)^d\; ,
\end{equation}
where $\tau_\gamma \phi(x,\xi) := \phi (x, \xi + \gamma)$. But, since by \eqref{opall3} $e^{i\ep D_xD_\xi}$
is a convolution operator, it commutes with all translations, which shows \eqref{Ltaukomm}.

Thus it remains to show that $R_N(b)$ is in $S_\delta^{r + N(1-2\delta)}(m)(\R^{2d})$ -
for then
it is in $S_\delta^{r + N(1-2\delta)}(m)(\R^d \times \T^d)$ -
and depends only on Fr\'echet-seminorms
$\|b\|_{\alpha,\beta}$ with $|\alpha|,|\beta|\geq N$.
We sketch the proof of these statements, since the standard proofs of \eqref{eDDent}
for $b\in S_\delta^r(m)(\R^{2d})$ or some similar classes (see e.g. Dimassi-Sj\"ostrand \cite{dima}, Grigis-Sj\"ostrand \cite{grisjo}, 
Martinez \cite{martinez}) do not directly lead to these more refined remainder
estimates.

First one proves the statement for $b\in {\mathscr C}_0^\infty (\R^{2d})$. Then the integral in \eqref{opall3}
converges absolutely and
\begin{equation}\label{RNformel}
R_N(b) = \frac{1}{(\ep 2\pi)^d} \int_{\R^{2d}} \left(e^{-\frac{i}{\ep} z\eta}
- \sum_{j=0}^{N-1} \frac{1}{j!} (-\frac{i}{\ep}z\eta)^j\right)
b(x-z, \xi-\eta)\, dz\,d\eta \; .
\end{equation}
Thus it is formally obvious that $\|R_N\|_{\alpha,\beta}$ depends only on Fr\'echet-seminorms
$\|b\|_{\alpha',\beta'}$ for $|\alpha'|,|\beta'| \geq N$.
One needs integration by parts and standard arguments to show that
$R_N\in S_\delta^{r + N(1-2\delta)}(m)(\R^d\times \T^d)$ and that
\eqref{RNnorm} holds
with constants $C_{\alpha,\beta,\alpha',\beta'}$ independent of $b$.

Now let $b\in S_\delta^r (m)(\R^{2d})$ and choose a cut-off function $h\in{\mathscr C}_0^\infty (\R^{2d})$ with $h=1$ on the ball with radius 1 and
$\supp h$ contained in the ball with radius 2. Set $h_R(\cdot) := h(\frac{\cdot}{R})$ and $b_R:= h_R b$ for
$R>1$. One readily verifies that the family $b_R$ is uniformly bounded in $S_\delta^r (m)(\R^{2d})$. By
standard arguments it follows that $b_R$ converges to $b$ in the topology of $S_\delta^r (\tilde{m})(\R^{2d})$
for $\tilde{m}(x,\xi) = \langle x \rangle \langle \xi\rangle m(x,\xi)$ (see e.g. Grigis-Sj\"ostrand \cite{grisjo}).
Furthermore, $R_N(b_R)$ is uniformly bounded in $S_\delta^{r + N(1-2\delta)}(m)(\R^{2d})$ - using the dominated
convergence theorem after integration by parts and the fact that
$R_N(D_x^\alpha D_\xi^\beta b) = D_x^\alpha D_\xi^\beta R_N (b)$ for all symbols $b$ -
and converges pointwise to some symbol $r_N\in S_\delta^{r + N(1-2\delta)}(m)(\R^{2d})$. Again, $R_N(b_R)$
converges to $r_N$ in the topology of $S_\delta^{r + N(1-2\delta)} (\tilde{m})(\R^{2d})$. Using
the continuity of $e^{i\ep D_xD_\xi}: S_\delta^r (\tilde{m})(\R^{2d}) \rightarrow S_\delta^r (\tilde{m})(\R^{2d})$,
it follows that
\[ e^{i\ep D_xD_\xi} b = \sum_{j=0}^{N-1} \frac{(i\ep D_x\cdot D_\xi)^j}{j!}\,b + r_N \; . \]
Thus $r_N = R_N(b)\in S_\delta^{r + N(1-2\delta)}(m)(\R^{2d})$, which completes the proof of Lemma \ref {A(D)}.
\end{proof}

\begin{rem}
The rougher standard argument in e.g. Grigis-Sj\"ostrand \cite{grisjo} and Martinez \cite{martinez}
splits $b=hb + (1-h)b$ with $h\in\Ce_0^\infty(\R^{2d})$. By stationary phase,
$e^{i\ep D_xD_\xi} (1-h) b \in S_\delta^{-\infty}(m)(\R^{2d})$, but its Fr\'echet-seminorms
depend on all Fr\'echet-seminorms of $b$! Of course, the relevant terms in the estimate for
$e^{i\ep D_xD_\xi} (1-h) b$ are precisely cancelled by corresponding terms in the estimate for
$e^{i\ep D_xD_\xi} h b$. This cancellation, however, is not evident from the estimates stated
\cite{grisjo} and \cite{martinez}.
\end{rem}

The following corollary concerns the composition of symbols.

\begin{cor}\label{a1a2}
The map
\[ {\mathscr C}_0^\infty \left(\R^d\times\T^d\right)\times{\mathscr C}_0^\infty
\left(\R^d\times\T^d\right) \ni
(a,b)\longmapsto a\# b\in{\mathscr C}^\infty \left(\R^d\times\T^d\right) \]
with
\begin{equation}\label{akreuzb2}
(a\# b) (x,\xi;\ep) :=
\left(e^{-i\ep D_y\cdot D_\xi}a(x,\xi;\ep)b(y,\eta;\ep)\right)|_{\natop{y=x}{\eta=\xi}}
\end{equation}
has a bilinear continuous extension :
\[ S_{\delta_1}^{r_1}(m_1)\left(\R^d\times\T^d\right)\times
S_{\delta_2}^{r_2}(m_2)\left(\R^d\times\T^d\right)\rightarrow
S_\delta^{r_1+r_2}(m_1m_2)\left(\R^d\times\T^d\right)\] for all
$\delta_k\in[0,\frac{1}{2}], k=1, 2$ and all order functions $m_1, m_2$,
where $\delta:=\max\{\delta_1, \delta_2\}$. If in addition $\delta_1 + \delta_2 <1$,
\begin{equation}\label{akreuzb}
(a\# b) (x,\xi;\ep)\sim \sum_{\alpha\in\N^d}
\frac{(i\ep)^{|\alpha|}}{|\alpha|!}\left(\partial_\xi^\alpha
a(x,\xi;\ep)\right)
\left(\partial_x^\alpha b(x,\xi;\ep)\right)
\end{equation}
(with respect to $k_N=r_1 + r_2 + N(1-\delta_1-\delta_2)$).
Writing
\[ a \# b(x,\xi;\ep) = \sum_{|\alpha|=0}^{N-1}
\frac{(i\ep)^{|\alpha|}}{|\alpha|!}\left(\partial_\xi^\alpha
a(x,\xi;\ep)\right)
\left(\partial_x^\alpha b(x,\xi;\ep)\right) + R_N(a,b,\ep)\, ,\]
the remainder $R_N$ is an element of the
symbol class $S_\delta^{r_1 + r_2 + N(1-\delta_1 - \delta_2)}(m_1m_2)$ and it depends linearly on a finite number
of derivatives of the single symbols $a$ and $b$. Furthermore it depends only on derivatives of $a$ and $b$ with respect
to $\xi$ and $x$ respectively which are at least of order $N$.
\end{cor}

\begin{proof}
By the Leibnitz rule, the map
\[ S_{\delta_1}^{r_1}(m_1)\times S_{\delta_2}^{r_2}(m_2)\ni (a,b)\mapsto
a\cdot b\in S_\delta^{r_1+r_2}(m_1m_2)\]
is continuous,
since each Fr\'echet-norm of the product depends only on a finite number of
Fr\'echet-norms of $a$ and $b$.
The same is true for the restriction map. The main part follows from Lemma \ref{A(D)} by
doubling the dimension of the space.
\end{proof}

It is shown in \cite{grisjo} that the $\#$-product of symbols reflects the composition of the
associated operators.
In particular
for $a\in S_{\delta_1}^{r_1}(m_1)\left(\R^d\times\T^d\right)$, $b\in
S_{\delta_2}^{r_2}(m_2)\left(\R^d\times\T^d\right)$ with $0\leq\delta_k \leq \frac{1}{2}, k=1, 2$,
\begin{equation}\label{AB}
\left(\Op_\ep^{\T}(a)\right)\circ \left(\Op_\ep^{\T}(b)\right) = \Op_\ep^{\T}(a\# b)\; .
\end{equation}

The following proposition is an adapted version of the Calderon-Vaillancourt-Theorem (see
Calderon-Vaillancourt \cite{calderon}) . The proof is inspired by the proof
of the Calderon-Vaillancourt-Theorem given by Hwang \cite{hwang}.

\begin{prop}\label{cald}
Let $a\in S^r_\delta(1)\left(\R^d\times\T^d\right)$
with  $0\leq \delta \leq \frac{1}{2}$.
Then there exists a constant $M>0$ such that, for the associated operator $\Op_\ep^{\T}(a)$ given by
\eqref{psdo2} the
estimate
\begin{equation}\label{cvab}
\left\| \Op_\ep^{\T}(a) u \right\|_{\ell^2(\disk)} \leq M \ep^r \|u\|_{\ell^2(\disk)}
\end{equation}
holds for any $u\in s\left(\disk\right)$ and $\ep>0$.
$\Op_\ep^{\T}(a)$ can therefore be extended to a continuous operator:
$\ell^2\left(\disk\right)\longrightarrow \ell^2\left(\disk\right)$ with
$\|\Op_\ep^{\T}(a)\| \leq M\ep^r$. Moreover $M$ can be chosen depending only on a finite number of
Fr\'echet-seminorms of the symbol $a$.
\end{prop}

\begin{rem}
There is a dual approach to the operators $\Op_\ep^\T(a)$, starting from pseudo-differential
calculus on the torus $\T^d = \R^d/2\pi\Z^d$ (see e.g. G\'erard-Nier \cite{gerard}). 
We denote by $\mathfrak{j}: \bigcup_{k\in \Z} H^k (\T^d) \rightarrow \mathscr{S}'(\R^d)$ the injection defined by periodic 
continuation, where $H^k(\T^d)$ is the Sobolev-space of order $k$ on the torus. Then we define the
$\ep$-quantization of a periodic symbol $a$,  i.e. $a(k + \mu, \eta) = a(k,\eta)$ for all $\mu\in 2\pi\Z^d$, in some H\"ormander class $S(m,g)$ by
\begin{equation}\label{re0} 
\Op_{\ep, t}^{\T^*} (a) = \mathfrak{j}^{-1}\circ \Op_{\ep, t} (a) \circ \mathfrak{j} \, ,
\end{equation}
where $\Op_{\ep, t}(b) : \mathscr{S}'(\R^d)\rightarrow \mathscr{S}'(\R^d)$ is induced from
\[ \Op_{\ep, t} (b) u (x) := (\ep 2\pi)^{-d} \iint_{\R^{2d}} e^{-\frac{i}{\ep}(k'-k)\eta}
b(tk + (1-t)k', \eta) u(k') \,dk'\, d\eta\, , \quad u\in\mathscr{S}(\R^d)  \]
(cf. Robert \cite{robert} and Dimassi-Sj\"ostrand \cite{dima}). $\Op_{\ep, t}^{\T^*} (a)$ is well defined, since by
the periodicity of $a$, the operator $\Op_{\ep, t}(a)$ commutes with all translations
$\tau_\gamma\, , \; \gamma\in \T^d$. Essentially, this is the approach in G\'erard-Nier \cite{gerard}. 
One now observes that \eqref{opall1} may be rewritten as
\begin{equation}\label{re1}
 F_\ep \circ \mathfrak{j} =\ep^d\,  r^* \circ \mathscr{F}_\ep \; ,\quad\text{for}\;\; r^* (u) = 
\sum_{\gamma\in\disk}u(\gamma) \delta_\gamma\, , \quad u\in s^*(\disk)
\end{equation}
where $r^*$ is the adjoint of the restriction map $r:\mathscr{S}(\R^d) \rightarrow s(\disk)$. 
Furthermore, a straightforward calculation gives
\begin{equation}\label{re2}
 \Op_\ep (b)\circ  F_\ep = F_\ep  \circ \Op_{\ep, 0} (\hat{b})\, , \quad \hat{b}(\xi, x) := b(x, \xi)\; .
\end{equation}
Thus, for $b\in S_\delta^r (m)(\R^d\times \T^d)$, the symbol $\hat{b}$ is periodic in the sense mentioned before
\eqref{re0}.
Moreover, taking adjoints in \eqref{Opasum2} gives on $s^*(\disk)$
\begin{equation}\label{re3}
 r^* \circ \Op_\ep^\T (a) = \Op_\ep (a) \circ  r^*
\end{equation}
for all $a\in S_\delta^r(m)(\R^d\times \T^d)$, since $(\Op_\ep^\T (a))^* = \Op_\ep^\T (a^\#)$ for 
$a^\#(x,\xi) = e^{i\ep D_xD_\xi} \bar{a}(x,\xi)$. By Lemma \ref{A(D)}, $a^\#\in S_\delta^r(m)(\R^d\times \T^d)$
for $a$ in this class, if $\delta\leq \frac{1}{2}$. 
Combining \eqref{re0}, \eqref{re1}, \eqref{re2} and \eqref{re3}
gives for $a\in S_\delta^r(m)(\R^d\times \T^d)$
\begin{equation}
 \Op_\ep^\T (a) \circ  \mathscr{F}_\ep = \mathscr{F}_\ep \circ  \Op_{\ep, 0}^{\T^*} (\hat{a})\, , 
\end{equation}
since $r^*$ is injective. 
Since $\mathscr{F}_\ep$
is unitary, $\ell^2(\disk)$-boundedness of $\Op_{\ep}^\T(a)$ is equivalent to
$L^2(\T^d)$-boundedness of $\Op_{\ep,0}^{\T^*}(\hat{a})$.

Under the additional assumption that
\[ \left| \partial_\eta^\alpha a(k,\eta)\right| \leq C_\alpha
\langle \eta \rangle^{m-|\alpha|}\; , \]
$L^2(\T^d)$-boundedness of $\Op_{\ep, t}^{\T^*}(a)$ follows from the
standard Calderon-Vaillancourt-Theorem for $\Op_\ep(a)$ in $\R^d$
and integration by parts
(see G\'erard-Nier \cite{gerard} for a simple proof in the case $\ep=1$, $t=1$; the proof works for any $t\in [0,1]$).
\end{rem}

\begin{proof}
Since $\ep^{-r} a\in S^0_\delta(1)$, we can restrict the proof to the case $r=0$.
It suffices to show that for all $u,v$ with compact support the estimate
\begin{equation}\label{cald1}
\left|\skpd{u}{\Op_\ep^{\T}(a) v}\right| \leq M  \| u \|_{\ell^2} \| v \|_{\ell^2}
\end{equation}
holds, where $M$ depends only on a finite number of Fr\'echet-seminorms
$\| \partial_x^\alpha \partial_\xi^\beta a\|_{\infty}$.
We assume that $\supp a$ is compact. The general case then follows by standard techniques
(approximating $a$ by a compactly supported sequence $a_n$ with $\Op_\ep^{\T}(a_n) \ra \Op_\ep^{\T}(a)$ strongly).
We have
\begin{multline}\label{skalarp}
\skpd{u}{\Op_\ep^{\T}(a) v}  \\
=(2\pi)^{-\frac{3d}{2}}\!\!\!\!\sum_{x\in\disk}\int\limits_{[-
\pi,\pi]^d}\!\!\!\!\! d\eta e^{-\frac{i}{\ep}\eta
x}\left({\mathscr F_\ep}^{- 1}\bar{u}\right)(\eta)\!\!\!\!
\sum_{y\in\disk}\int\limits_{[-\pi,\pi]^d}
\!\!\!d\xi\, e^{-\frac{i}{\ep}(x-y)\xi}a(x,\xi;\ep)v(y)  \; .
\end{multline}
By the assumption on $a$, the iterated integrals (and sums) in \eqref{skalarp} can be understood as
integrals on the product space (thus Fubini`s Theorem holds).
Let $\zeta\in\Ce_0^\infty(\R)$ be a cut-off-function with $\zeta=1$ at zero, then we split the right
hand side of \eqref{skalarp} in two summands
by introducing $\zeta (\frac{1}{\sqrt{\ep}}|\xi+\eta|)$ and $(\id-\zeta)(\frac{1}{\sqrt{\ep}}|\xi + \eta|)$.
It then suffices to show that the part multiplied by $(\id - \zeta)(\frac{1}{\sqrt{\ep}}|\xi+\eta|)$,
which we denote by $I_1$, is an element of $S^\infty(1)$ and the part multiplied with
$\zeta (\frac{1}{\sqrt{\ep}}|\xi+\eta|)$, denoted by $I_2$, is bounded by a constant independent of $\ep$.

To analyze $I_1$, we use the operators
\begin{equation}\label{L1L2}
L_1  :=
\frac{\id - \ep \Delta_\xi}{\left\langle \tfrac{1}{\sqrt{\ep}}(x-y)\right\rangle^2}
\quad\text{and}\quad L_2  :=
\frac{- \Delta_x^{\sqrt{\ep}}}{2d - 2\sum_\nu\cos
\left(\tfrac{1}{\sqrt{\ep}}(\xi_\nu + \eta_\nu)\right)}\, ,
\end{equation}
where $-\Delta_x^{\sqrt{\ep}} := 2d - \sum_\nu \left(\tau_{\sqrt{\ep}e_\mu} + \tau_{-\sqrt{\ep}e_\mu}\right)$ is a scaled version
of the discrete Laplacian $\Delta_\ep$ defined in \eqref{diskretlap}.
Then $L_1$ and $L_2$ leave $e^{-\frac{i}{\ep}((x-y)\xi + x\eta)}$
invariant, and we have by the symmetry of $\Delta_x^{\sqrt{\ep}}$ and $\id - \ep\Delta_\xi$ (using Fubini)
\begin{eqnarray*}
I_1 &=& (2\pi)^{-\frac{3d}{2}}\!\!\!\!\!\sum_{y,x\in\disk}
\int\!\!\!\!\int\limits_{[-\pi,\pi]^d}\!\!\!
\!\!d\eta d\xi\,
\left( L_2^kL_1^l e^{-\frac{i}{\ep}((x-y)\xi+x\eta)}\right)\\
&&\hspace{0.3cm}\times \left({\mathscr
F_\ep}^{-1}\bar{u}\right)(\eta) (\id - \zeta)
\left(\tfrac{1}{\sqrt{\ep}}|\xi+\eta|)\right) a(x,\xi;\ep) v(y)  \\
&=& (2\pi)^{-\frac{3d}{2}}\sum_{y,x\in \disk}\int\!\!\!
\int_{[-\pi,\pi]^d}\!\!\!d\eta\,
d\xi\,
 e^{-\frac{i}{\ep}((x-y)\xi+ x \eta)}
\left({\mathscr F_\ep}^{-1}\bar{u}\right) (\eta)\frac{v(y)}
{\left\langle \tfrac{1}{\sqrt{\ep}}(x-y)\right\rangle^{2l}} \\
&&\hspace{0.1cm}\times
\sum_{|\alpha|\leq 2l} K_{k,\alpha}(\xi,\eta;\ep)
 P_\alpha (\sqrt{\ep}D_\xi)\left(-\Delta_x^{\sqrt{\ep}}\right)^k
a(x,\xi;\ep) \, ,
\end{eqnarray*}
where
\[ K_{k,\alpha}(\xi,\eta; \ep) = Q_{\alpha} (\sqrt{\ep}D_\xi)\frac{\id - \zeta\left(\tfrac{1}{\sqrt{\ep}}|\xi+
\eta|\right)}
{\left(2d - 2\sum_\nu\cos
\left(\tfrac{1}{\sqrt{\ep}}(\xi_\nu + \eta_\nu)\right)\right)^k}
\]
and $Q_{\alpha}$ and $P_\alpha$ denote polynomials with $\deg Q_\alpha + \deg P_\alpha = 2l$.
With the notation
\begin{eqnarray}
G_{k,\alpha}(x,\xi;\ep)  &:=&
\mathscr{F}_\ep \left[\left({\mathscr F_\ep}^{-1}\bar{u}\right)(~\cdot~)
K_{k,\alpha}(\xi,.;\ep)\right](x)
 \label{Glkalpha} \\
F_l(x,\xi;\ep) &=& \mathscr{F}_\ep^{-1}\left[ \frac{v(~\cdot~)}{\left\langle
\tfrac{1}{\sqrt{\ep}}(x-~\cdot~) \right\rangle^{2l}} \right](\xi) \label{F_l}
\end{eqnarray}
we have
\begin{equation}\label{Ieinsgleich}
I_1  =(2\pi)^{-\frac{d}{2}}\sum_{x\in\disk}
\int_{[-\pi,\pi]^d}\!\!\!d\xi\,
 e^{-\frac{i}{\ep}x\xi}F_l(x,\xi)
\sum_{|\alpha|\leq 2l} G_{k,\alpha}(x,\xi)  P_\alpha (\sqrt{\ep}D_\xi)\left(-\Delta_x^{\sqrt{\ep}}\right)^k
a(x,\xi;\ep) \, .
\end{equation}
Thus by the Schwarz-inequality
\begin{equation}\label{abIeinsa}
|I_1| \leq \sum_{|\alpha|\leq 2l} \sup_{x,\xi}\left|P_\alpha (\sqrt{\ep}D_\xi) \left(-\Delta_x^{\sqrt{\ep}}\right)^k
a(x,\xi;\ep)\right|\; \| F_l \|_{\ell^2\times \T^d}^2\,
\| G_{k,\alpha}\|^2_{\ell^2\times\T^d}\; .
\end{equation}
By the isometry of $\mathscr{F}_{\ep}$
\begin{align}\label{normFl}
\| F_l \|_{\ell^2\times \T^d}^2
& = \| \mathscr{F}_\ep F_l \|_{\ell^2\times \ell^2}^2\\
&\leq \sum_{x,y\in\disk}
 |v(y)|^2\left\langle \frac{x-y}{\sqrt{\ep}} \right\rangle^{-4l}
 \leq \| v \|^2_{\ell^2(\disk)} \sum_{t\in\disk} \left(1+ \frac{t^2}{\ep}\right)^{-2l}\nonumber\\
&\leq
C_l \ep^{-\frac{d}{2}} \| v \|^2_{\ell^2(\disk)}\nonumber
\end{align}
with $t = x-y$, where the last estimate follows from \eqref{lemma2.10.2} for $l$ big enough.
For $G_{k,\alpha}$, we have by the isometry of ${\mathscr F}^{-1}_\ep$
\begin{multline}\label{Glkalphaab1}
\| G_{k,\alpha}\|^2_{\ell^2\times\T^d}
= \|\mathscr{F}^{-1}_\ep G_{k,\alpha}\|^2_{\T^d\times\T^d} \\
\leq \int_{[-\pi,\pi]^d}\!\!\!d\eta \left|\left({\mathscr
F_\ep}^{-1}\bar{u}\right)(\eta)\right|^2
\int_{[-\pi,\pi]^d}\!\!\!d\xi \left| K_{k,\alpha}(\xi,\eta;\ep)\right|^2
 \, .
\end{multline}
Using $\sqrt{\ep}D_\xi f(\frac{\xi}{\sqrt{\ep}}) = O(1)$ for any smooth function with bounded
derivative and $\pi^2(1-\cos(\frac{\tau}{\sqrt{\ep}})) \geq
\frac{\tau^2}{\ep}$ for $|\frac{\tau}{\sqrt{\ep}}|\leq\pi$, we have for $\tau = \xi + \eta$
\begin{equation}\label{Qab}
\left| K_{k,\alpha}(\xi,\eta;\ep) \right|^2
\leq \tilde{C}_{k,\alpha} \left| \frac{\tau}{\sqrt{\ep}}\right|^{-4k}\; .
\end{equation}
Since for $k$ large enough
\[ \int_{\supp (\id-\zeta)(\frac{\cdot}{\sqrt{\ep}})}\left|\frac{\tau}{\sqrt{\ep}}\right|^{-4k} \,d\tau \leq C_k
\ep^{\frac{d}{2}} \; , \]
we get by inserting \eqref{Qab} into \eqref{Glkalphaab1}
\begin{equation}\label{Glkalphaab}
\| G_{k\alpha}\|^2_{\ell^2\times\T^d} \leq
C_{k,\alpha} \ep^{\frac{d}{2}} \| u \|^2_{\ell^2(\disk)} \; .
\end{equation}
To analyze $P_\alpha (\sqrt{\ep}D_\xi)\left(-\Delta_x^{\sqrt{\ep}}\right)^k
a(x,\xi;\ep)$, we use that by Taylor expansion
\[ \Delta_x^{\sqrt{\ep}} a(x,\xi;\ep) = -\sum_{\nu=1}^d \left[\ep \partial_{x_\nu}^2 a(x,\xi;\ep) +
\frac{\ep^{\frac{3}{2}}}{3!} \int_0^1 \partial_{x_\nu}^3 a(x+t\sqrt{\ep} e_\nu)\, dt \right]\; .\]
By iteration, we have for $a\in S_\delta^0(1)$
\begin{equation}\label{asup}
\sup_{x,\xi}\left|P_\alpha (\sqrt{\ep}D_\xi) \left(-\Delta_x^{\sqrt{\ep}}\right)^k
a(x,\xi;\ep)\right|
 \leq \tilde{M}_{k,\alpha} \ep^{2k(\frac{1}{2} - \delta)} \, ,
\end{equation}
where $\tilde{M}_{k,\alpha}$ depends only on Fr\'echet-seminorms of $a$ up to order $3k+|\alpha|$.
Inserting \eqref{normFl}, \eqref{Glkalphaab} and \eqref{asup} in \eqref{abIeinsa} yields for any $k\in\N$
\begin{equation}\label{normIeins}
|I_1| \leq  M \ep^{2k(\frac{1}{2}-\delta)} \| u \|_{\ell^2(\disk)}\| v \|_{\ell^2(\disk)} \; .
\end{equation}
Thus $I_1= O(\ep^\infty)$.

To get an estimate for the modulus of $I_2$, which denotes the integral over the support of
$\zeta$, we
use $L_1$ given in \eqref{L1L2} to get by integration by parts and similar arguments
\begin{multline*}
I_2 =
(2\pi)^{-\frac{3d}{2}}\!\!\!\sum_{x,y\in\disk}
\frac{v(y)}{\left\langle\frac{1}{\sqrt{\ep}} ( x- y) \right\rangle^{2l}}
\int\!\!\!
\int_{[-\pi,\pi]^d}\!\!\!d\eta d\xi\,
 e^{-\frac{i}{\ep}((x-y)\xi+x\eta)}\times\\
\times \left({\mathscr F_\ep}^{-1}\bar{u}\right)(\eta)
 \left(\id - \ep \Delta_\xi\right)^l \zeta \left(\tfrac{1}{\sqrt{\ep}}|\xi+\eta|\right)
a(x,\xi;\ep)\, .
\end{multline*}
Setting, for $P_\alpha$ and $Q_\alpha$ as above,
\[
G_{\alpha}(x, \xi;\ep) :=  \mathscr{F}_\ep
\left[ \left(\mathscr{F}_\ep^{-1}\bar{u}\right)(~\cdot~)
 Q_\alpha(\sqrt{\ep}D_\xi)\zeta\left(\tfrac{1}{\sqrt{\ep}}|\xi +~\cdot~|\right)  \right](x)
\]
we have, with $F_l$ as in \eqref{F_l},
\begin{equation}\label{Izweigleich}
I_2 = (2\pi)^{-\frac{d}{2}}\sum_{x\in\disk}\int_{[-\pi,\pi]^d}\!\!\!
d\xi e^{-\frac{i}{\ep}x \xi} F_l(x,\xi;\ep) \sum_{|\alpha|\leq 2l} G_{\alpha}(x,\xi;\ep) P_\alpha(\sqrt{\ep}D_\xi)a(x,\xi;\ep)  \,.
\end{equation}
By the isometry of ${\mathscr F}_\ep$ and the arguments leading to \eqref{Glkalphaab}, we have
\begin{align}\label{Glab}
\| G_{\alpha} \|^2_{\ell^2\times\T^d}
&= \| \mathscr{F}_\ep^{-1} G_{\alpha} \|^2_{\T^d\times\T^d}\nonumber \\
&\leq
 \int_{[-\pi,\pi]^d}\!\!d\eta \left|\left({\mathscr F_\ep}^{-1}\bar{u}\right)(\eta)\right|^2
 \int_{[-\pi,\pi]^d}\!\!\!
d\xi \left|Q_\alpha(\sqrt{\ep}D_\xi)\zeta\left(\tfrac{1}{\sqrt{\ep}}|\xi + \eta |\right) \right|^2 \nonumber\\
&\leq C_l \| u \|^2_{\ell^2(\disk)}
\int_{\supp\zeta(\frac{\cdot}{\sqrt{\ep}})} d\xi \;\leq \; \ep^{\frac{d}{2}}C_l \| u \|^2_{\ell^2(\disk)}\, ,
\end{align}
where the estimate in the last line follows from the scaling of $\zeta$.
Analog to \eqref{abIeinsa} we get by \eqref{Glab}, \eqref{normFl} and \eqref{asup} for $k=0$
\[ |I_2| \leq M  \| u \|_{\ell^2(\disk)}\| v\|_{\ell^2(\disk)} \]
and therefore we finally get \eqref{cald1}.
\end{proof}

For $a\in S_{\delta_a}^{r_a}(m_a)(\R^d\times \T^d)$ and $b\in S_{\delta_b}^{r_b}(m_b)(\R^d\times \T^d)$
let $[a,b]_{\#}:= a\# b - b\# a$
denote the commutator in symbolic calculus.
Then by \eqref{AB}
\begin{equation}\label{komAB}
\Op_\ep^{\T}\left( [a,b]_\# \right) = \left[ \Op_\ep^{\T}(a), \Op_\ep^{\T}(b)\right]\; .
\end{equation}
The following lemma, which gives the resulting symbol class
of double commutators, is an application
of Corollary \ref{a1a2} and \ref{AB}.

\begin{Lem}\label{kommut}
Let $h(x,\xi)\in S_{\delta_2}^{r_2}(m_2)(\R^d\times \T^d), \delta_2<\frac{1}{2}$ and let
$\chi, \phi\in S_{\delta_1}^{r_1}(m_1)(\R^d\times \T^d),\, \delta_1 < \frac{1}{2}$,
where $\chi$ does not depend on $\xi$ and $\phi$ does not depend on $x$.
Then for $\alpha,\alpha_1,\alpha_2\in\N^d$ with $\alpha_1+\alpha_2 = \alpha$ and $|\alpha_k|\geq 1, k=1,2$,
for $\delta := \max \{\delta_1,\delta_2\}$ and for any $N\in\N, N\geq 3$:
\begin{enumerate}
\item $[\chi,[\chi, h]_\#]_\#\in S_\delta^{2-2(\delta_1 + \delta_2)}(m_1^2m_2)$
and it has the expansion
\[ [\chi(x),[\chi(x), h(x,\xi)]_\#]_\# = \sum_{2\leq |\alpha|< N}\frac{(i\ep)^{|\alpha|}}{|\alpha|!}
\left(\partial_\xi^\alpha h\right)(x,\xi)\sum_{\alpha_1, \alpha_2}
\left(\partial_x^{\alpha_1}\chi\right)(x) \left(\partial_x^{\alpha_2}\chi\right)(x) + R_N\, .\]
\item $[\phi(\xi),[\phi(\xi),h(x,\xi)]_\#]_\#\in S_\delta^{2-2(\delta_1 + \delta_2)}(m_1^2m_2)$
and it has the expansion
\[ [\phi(\xi),[\phi(\xi),h(x,\xi)]_\#]_\# = \sum_{2\leq |\alpha|<N}\frac{(i\ep)^{|\alpha|}}{|\alpha|!}
\left(\partial_x^\alpha h\right)(x,\xi)\sum_{\alpha_1, \alpha_2}
\left(\partial_\xi^{\alpha_1}\phi\right)(\xi) \left(\partial_\xi^{\alpha_2}\phi\right)(\xi) + \tilde{R}_N\,.\]
\end{enumerate}
where $R_N$ and $\tilde{R}_N$ are elements of the symbol class
$S_\delta^{N(1-\delta_1 - \delta_2)}(m_1^2m_2)$ and depend linearly on a finite number of
Fr\'echet-seminorms of the single symbols. Furthermore they depend
only on the derivatives of $h$, which are at least of order $N$ and of the product of derivatives of $\chi$ and
$\phi$ respectively, which are of order $N_1$ and $N_2$, such that $N_1 + N_2 \geq N$.
\end{Lem}

\begin{proof}
(a):\\
The double commutator is given by
\begin{equation}\label{dkchi} [\chi(x),[\chi(x), h(x,\xi)]_\#]_\#  = \chi \# \chi \# h (x,\xi) + h \# \chi \# \chi (x,\xi) -
2 \chi \# h \# \chi (x,\chi) \, .
\end{equation}
By Lemma \ref{a1a2}, these terms are for $alpha\in\N^d$ given by
\begin{align*}
\chi \# \chi \# h (x,\xi) &= \chi \cdot \chi \cdot h (x,\xi) \\
h \# \chi \# \chi (x,\xi) &= \sum_{|\alpha|\leq N-1} \frac{(i\ep)^{|\alpha|}}{|\alpha|!}
\left(\partial_\xi^\alpha h\right)\left(\partial_x^\alpha \chi^2\right) (x,\xi) + R_N(x,\xi; \ep \\
\chi \# h \# \chi (x,\chi) &=  \sum_{|\alpha|\leq N-1} \frac{(i\ep)^{|\alpha|}}{|\alpha|!}
\chi \left(\partial_\xi^\alpha h\right)\left(\partial_x^\alpha \chi\right) (x,\xi) + \tilde{R}_N(x,\xi; \ep)   \, .
\end{align*}
where $R_N, \tilde{R}_N \in S_\delta^{N(1-\delta_1-\delta_2)}(m_1^2 m_2)$.
The terms with $|\alpha|=0$ and $|\alpha|=1$ cancel in (\ref{dkchi})
Furthermore all terms with $2 \chi_j\partial_x^\alpha \chi_j$ cancel.
Thus the Leibnitz formula
gives  the expansion
\[ [\chi(x),[\chi(x), h(x,\xi)]_\#]_\# = \!\! \sum_{2\leq |\alpha|\leq N-1}
(i\ep)^{|\alpha|}
\left(\partial_\xi^\alpha h\right) \sum_{\alpha_1,\alpha_2\in\N^d}
\frac{1}{|\alpha_1|! |\alpha_2|!}
\left(\partial_x^{\alpha_1} \chi\right)\left(\partial_x^{\alpha_2}\chi\right) (x,\xi) + R_N(x,\xi;\ep)  \]
where the second sum runs over $\alpha_1+\alpha_2 = \alpha$ with $|\alpha_k|\geq 1, k=1,2$  and $R_N\in S_\delta^{N(1-\delta_1-\delta_2)}(m_1^2 m_2)$.
The statement on the symbol class follows at once from this expansion, since each summand is at least of
order $\ep^{2(1-\delta_1 - \delta_2)}$ and by use of the Leibnitz rule.\\

(b):\\
As above the double commutator consists of the terms
\begin{equation}\label{dkphi}
[\phi(\xi),[\phi(\xi), h(x,\xi)]_\#]_\#  = \phi \# \phi \# h (x,\xi) + h \# \phi \# \phi (x,\xi) -
2 \phi \# h \# \phi (x,\chi) \,
\end{equation}
and the summands have for $\alpha\in\N^d$ the expansions
\begin{align*}
\phi \# \phi \# h (x,\xi) &= \sum_{|\alpha| \leq N-1} \frac{(i\ep)^{|\alpha|}}{|\alpha|!}
\left(\partial_x^\alpha h\right)\left(\partial_\xi^\alpha \phi^2\right) (x,\xi) + R_N(x,\xi;\ep)  \\
h \# \chi \# \chi (x,\xi) &= h \cdot \phi \cdot \phi (x,\xi)  \\
\chi \# h \# \chi (x,\chi) &=  \sum_{|\alpha| \leq N-1} \frac{(i\ep)^{|\alpha|}}{|\alpha|!}
\phi \left(\partial_x^\alpha h\right)\left(\partial_\xi^\alpha \phi\right) (x,\xi) + \tilde{R}_N(x,\xi;\ep)  \; ,
\end{align*}
where $R_N, \tilde{R}_N \in S_\delta^{N(1-\delta_1-\delta_2)}(m_1^2 m_2)$.
Therefore, as discussed in (a), \eqref{dkphi} gives with 
\[ [\phi(\xi),[\phi(\xi), h(x,\xi)]_\#]_\#  \sim \sum_{\natop{\alpha\in\N^d}{2\leq |\alpha|\leq N-1}}
\frac{(i\ep)^{|\alpha|}}{|\alpha|!}
\left(\partial_x^\alpha h\right)\sum_{\natop{\alpha_1,\alpha_2\in\N^d}{\alpha_1 + \alpha_2 = \alpha}}
\left(\partial_\xi^{\alpha_1} \phi\right) \left(\partial_\xi^{\alpha_2} \phi\right)(x,\xi) + R_N(x,\xi;\ep)\]
with $\alpha_1+\alpha_2=\alpha$, $|\alpha_k|\geq 1, k=1,2$ and $R_N\in S_\delta^{N(1-\delta_1-\delta_2)}(m_1^2 m_2)$.
The statement on the symbol class follows from this expansion as discussed in
(a).\\
The additional properties of $R_N$ and $\tilde{R}_N$ respectively follow immediately from the properties of remainder
in Corollary \ref{a1a2}.
\end{proof}

\section{Persson's Theorem in the discrete setting}\label{persson}

In this section we will prove a theorem on the infimum of the essential spectrum of $H_\ep$
acting in $\ell^2\left(\disk\right)$, which is similar to Persson's Theorem for
Schr\"odinger operators. The proof
follows the proof of Persson's Theorem in the Schr\"odinger setting given in Helffer \cite{helf2}
and Agmon \cite{agmon} respectively.

\begin{theo}\label{perssontheo}
Let $H_\ep = T_\ep + V_\ep$ satisfy Hypothesis \ref{hypsatz}, denote by $\sigma_{ess}(H_\ep)$
the essential spectrum of
$H_\ep$ and define
\begin{equation}\label{pers14}
\Sigma (H_\ep) := \sup_{\natop{K\subset \disk}{\text{finite}}} \inf \left\{
\frac{\skpd{H_\ep \phi}{\phi}}{\|\phi\|^2_{\ell^2}}\, |\, \phi\in c_0\left(\disk\setminus K\right)\right\}\; ,
\end{equation}
where $c_0(D)$ denote the space of real-valued functions on $\disk$ with compact, i.e. finite, support
in $(\disk \setminus D)$.
Then
\[ \inf \sigma_{ess} \left(H_\ep\right) = \Sigma \left(H_\ep\right) \; . \]
\end{theo}

The proof of Theorem \ref{perssontheo} is divided in two Lemmata and the main part.

\begin{Lem}\label{perslem1}
For $x\in\disk$ and $R>0$ let $B_x(R):= \{ y\in\disk\, |\, |x-y| < R\}$ denote the ball around
$x$ with radius $R$ and
\begin{equation}\label{pers13}
\Lambda_R(x,H_\ep) := \inf \left\{ \frac{\skpd{H_\ep \phi}{\phi}}{\|\phi\|_{\ell^2}^2}\, ;\,
\phi\in c_0\left(B_x(R)\right)\right\}
\; .
\end{equation}
Then for all $\delta>0$ there exists a radius $R_\delta >0$ such that for all $R>R_\delta$ and $\phi\in c_0\left(\disk\right)$
\[ \skpd{H_\ep \phi}{\phi} \geq \sum_{x\in\disk} \left(\Lambda_R(x,H_\ep) - \delta\right) |\phi(x)|^2 \;  . \]
\end{Lem}

\begin{proof}[Proof of Lemma \ref{perslem1}]
Let $\rho\in\Ce_0^\infty \left(\R^d\right)$ be real valued with $\rho(x) = 0$ for $|x| \geq \frac{1}{2}$ and
$\int_{\R^d} |\rho (x)|^2\, dx = 1$ and define
\[ \rho_{y,R}:= \rho\left(\frac{y-x}{R}\right)\; .\]
Then $\rho_{y,R}\phi  \in
c_0 \left(B_y\left(\frac{R}{2}\right)\right)$ and therefore by the definition of
$\Lambda_{R}$
\[ \skpd{H_\ep \rho_{y,R}\phi}{\rho_{y,R}\phi} \geq \Lambda_{\frac{R}{2}}(y,H_\ep)
\|\rho_{y,R}\phi\|^2_{\ell^2}\; . \]
Since $B_y\left(\frac{R}{2}\right)\subset B_x(R)$ for $|x-y|< \frac{R}{2}$
and thus $\Lambda_{\frac{R}{2}}(y) \geq \Lambda_R(x)$, we get the estimate
\begin{equation}\label{pers3}
\skpd{H_\ep \rho_{y,R}\phi}{\rho_{y,R}\phi} \geq \sum_{x\in\disk} \Lambda_R(x,H_\ep) (\rho_{y,R}\phi)^2(x)\; .
\end{equation}
To analyze the scalar product we use that $T_\ep$ is self adjoint and $\phi, \rho$ are real valued, yielding
\begin{multline*}
\skpd{T_\ep\rho_{y,R}\phi}{\rho_{y,R}\phi} = \frac{1}{2}\left(\skpd{T_\ep \rho_{y,R}\phi}{\rho_{y,R}\phi} +
\skpd{\rho_{y,R}\phi}{T_\ep\rho_{y,R}\phi}\right)  \\
=\frac{1}{2} \left(\skpd{T_\ep \phi}{\rho^2_{y,R}\phi} + \skpd{[T_\ep, \rho_{y,R}]\phi}{\rho_{y,R}\phi} +
\skpd{\rho^2_{y,R}\phi}{T_\ep \phi} + \skpd{\rho_{y,R}\phi}{[T_\ep,\rho_{y,R}]\phi}\right)  \\
= \skpd{T_\ep \phi}{\rho^2_{y,R}\phi} + \frac{1}{2} \left( \skpd{[T_\ep, \rho_{y,R}]\phi}{\rho_{y,R}\phi} +
\skpd{\rho_{y,R}\phi}{[T_\ep,\rho_{y,R}]\phi}\right) \; .
\end{multline*}
Since $[T_\ep, \rho_{y,R}]^* = - [T_\ep, \rho_{y,R}]$ it follows that
\[\skpd{T_\ep\rho_{y,R}\phi}{\rho_{y,R}\phi} = \skpd{T_\ep \phi}{\rho^2_{y,R}\phi} + \frac{1}{2}
\skpd{(\rho_{y,R}[T_\ep, \rho_{y,R}]- [T_\ep, \rho_{y,R}]\rho_{y,R})\phi}{\phi} \]
and since $V_\ep$ commutes with $\rho_{y,R}$, we therefore get
\begin{equation} \label{pers1}
\skpd{H_\ep \phi}{\rho^2_{y,R}\phi}  = \skpd{H_\ep\rho_{y,R}\phi}{\rho_{y,R}\phi}  -
\frac{1}{2} \skpd{[\rho_{y,R},[T_\ep, \rho_{y,R}]]\phi}{\phi}\; .
\end{equation}
To analyze the double commutator, we use the symbolic calculus introduced in
Section \ref{symbols}. By Lemma \ref{kommut},
the symbol associated to the operator $[\rho_{y,R},[T_\ep, \rho_{y,R}]]$ is given by
\begin{multline}\label{pers2}
\rho_{y,R}(x),[t(x,\xi), \rho_{y,R}(x)]_\#]_\# \\
 = \sum_{\natop{\alpha\in\N^d}{2\leq |\alpha|<N}}
\frac{(i\ep)^{|\alpha|}}{|\alpha|!}
\left(\partial_\xi^\alpha t\right)(x,\xi)\sum_{\natop{\alpha_1,\alpha_2}{|\alpha_1|+|\alpha_2|=|\alpha|}}
\left(\partial_x^{\alpha_1}\rho_{y,R}\right)(x) \left(\partial_x^{\alpha_2}\rho_{y,R}\right)(x) + R_N(t,\rho_{y,R})\, ,
\end{multline}
where $R_N$ depends of a finite number of derivatives of $\rho_{y,R}$, which are at least of order $N$.
By the scaling of $\rho_{y,R}$, it follows that $|\nabla_x \rho_{y,R}(x)| \leq \frac{C}{R}$ for $C$ suitable.
Since all terms in the finite sum in \eqref{pers2} and the remainder $R_N$ depend on a product of two
(at least first order) derivatives of
$\rho_{y,R}$, any Fr\'echet semi-norm of the symbol of the double commutator is of order $\frac{1}{R^2}$. By
Proposition \ref{cald}, the same statement follows for the operator-norm of the associated operator, thus there is a
constant $C>0$ such that
\begin{equation}\label{pers5}
\| [\rho_{y,R},[T_\ep, \rho_{y,R}]]\|_\infty \leq \frac{C}{R^2}
\end{equation}
By the Cauchy-Schwarz inequality, we get by inserting \eqref{pers3} and \eqref{pers5} in \eqref{pers1}
\begin{equation}\label{pers7}
\skpd{H_\ep \phi}{\rho^2_{y,R} \phi} \geq
\sum_{x\in\disk} \Lambda_R(x,H_\ep) |\rho_{y,R}\phi(x)|^2 - \frac{C}{R^2}\sum_{x\in B_y(R)}
|\phi(x)|^2\; .
\end{equation}
We remark that by setting $z=\frac{y-x}{R}$
\begin{equation}\label{pers8}
\int_{\R^d} \rho_{y,R}^2(x) \, dy = R^d \int_{\R^d} \rho^2 (z)\, dz = R^d
\end{equation}
and
\begin{equation}\label{pers9}
\int_{\R^d} \id_{\{|x-y|<R\} }\,dy = R^d \int_{\R^d}\id_{\{|z|<1\} }\,dz = C R^d\; .
\end{equation}
Thus integration of the left hand side of \eqref{pers7} with respect to $y$ yields by \eqref{pers8}
\begin{equation}\label{pers10}
\int_{\R^d}\skpd{H_\ep \phi}{\rho^2_{y,R} \phi}\, dy =
\skpd{H_\ep \phi}{\int_{\R^d}\rho^2_{y,R}\,dy \phi} =
R^d \skpd{H_\ep \phi}{\phi}\; .
\end{equation}
If we integrate the right hand side of \eqref{pers7} with respect to $y$ and use \eqref{pers9}, we get
\begin{multline}\label{pers11}
\int_{\R^d} \left(\sum_{x\in\disk} \Lambda_R(x,H_\ep) \rho_{y,R}^2(x) \phi^2(x) -
\frac{C}{R^2}\sum_{x\in \disk} \id_{\{|x-y|<R\}}|\phi(x)|^2\right)\, dy \\
= R^d\left( \sum_{x\in\disk} \Lambda_R(x,H_\ep) \phi^2(x) -
\frac{C'}{R^2}\sum_{x\in \disk} |\phi(x)|^2\right)\; .
\end{multline}
The Integration of both sides of \eqref{pers7} with respect to $y$ and division by $R^d$ gives by \eqref{pers10} and
\eqref{pers11}
\begin{equation}\label{pers12}
\skpd{H_\ep \phi}{\phi} \geq \sum_{x\in\disk} \left( \Lambda_R(x, H_\ep) - \frac{C}{R^2}\right)
|\phi(x)|^2\; .
\end{equation}
By choosing for $\delta >0$ the radius $R_\delta = \sqrt{\frac{C}{\delta}}$, the
statement of Lemma \ref{perslem1} follows for all $R>R_\delta$ by \eqref{pers12}.
\end{proof}

The family $\Lambda_R(x,H_\ep)$ describes the lowest eigenvalue of the Dirichlet problem with respect to
the ball $B_x(R)$. The next lemma relates this family with $\Sigma (H_\ep)$.

\begin{Lem}\label{perslem2}
Let $\Lambda_R(x,H_\ep)$ and $\Sigma (H_\ep)$ defined in \eqref{pers13} and \eqref{pers14} respectively, then
\begin{equation}\label{pers15}
\Sigma (H_\ep) = \lim_{R\to +\infty} \liminf_{|x|\to \infty} \Lambda_R(x,H_\ep)\; .
\end{equation}
\end{Lem}

\begin{proof}[Proof of Lemma \ref{perslem2}]

We split the proof in two parts showing the two fundamental inequalities.\\

{\sl Step 1:} Estimate from above
\begin{equation}\label{pers16}
\Sigma (H_\ep) \leq \lim_{R\to +\infty} \liminf_{|x|\to \infty} \Lambda_R(x,H_\ep)
\end{equation}

Let $K\subset\disk$ compact and $R>0$ fixed. Then $B_x(R) \subset \disk\setminus K$ for $|x|$ large enough and thus
\[ \inf\left\{ \frac{\skpd{H_\ep \phi}{\phi}}{\|\phi\|_{\ell^2}^2}\, ;\, \phi\in c_0\left(\disk\setminus K\right)\right\}  \leq
\inf \left\{
\frac{\skpd{H_\ep \phi}{\phi}}{\|\phi\|_{\ell^2}^2}\, ;\, \phi\in c_0\left(B_x(R)\right)\right\}
\left(=\Lambda_R(x, H_\ep)\right)\; . \]
This inequality is satisfied for all $|x|$ large enough and the left hand side is
independent of $x$, thus
\[ \inf\left\{ \frac{\skpd{H_\ep \phi}{\phi}}{\|\phi\|_{\ell^2}^2}\, ;\, \phi\in c_0\left(\disk\setminus K\right)\right\}  \leq
\liminf_{|x|\to\infty} \Lambda_R(x,H_\ep) \; . \]
The left hand side of this inequality is independent of $R$ and the right
hand side understood as a function in $R$ is monotonically decreasing and bounded from below, thus
the limit $R\to\infty$ is well defined and
\[ \inf\left\{ \frac{\skpd{H_\ep \phi}{\phi}}{\|\phi\|_{\ell^2}^2}\, ;\, \phi\in c_0\left(\disk\setminus K\right)\right\}  \leq
\lim_{R\to +\infty}\liminf_{|x|\to\infty} \Lambda_R(x,H_\ep) \; . \]
Now the right hand side is independent of the choice of $K$, thus we can take the supremum over all compact sets
$K\subset \disk$ and by the definition of $\Sigma (H_\ep)$, this shows \eqref{pers16}.\\

{\sl Step 2:}  Estimate from below
\begin{equation}\label{pers17}
\Sigma (H_\ep) \geq \lim_{R\to +\infty} \liminf_{|x|\to \infty} \Lambda_R(y,H_\ep)\; .
\end{equation}

By the definition of $\liminf$, it follows that for all $\delta>0$ and all $R>0$ there exists an $R_0$ such that
for all $|x|>R_0$
\[ \Lambda_R(x,H_\ep) \geq \liminf_{|x|\to\infty} \Lambda_R(x,H_\ep) - \delta \; . \]
It follows immediately that for all $\phi\in c_0\left(\disk\setminus \overline{B_0(R_0)}\right)$
\begin{equation}\label{pers18}
\sum_{x\in\disk} \Lambda_R(x,H_\ep) |\phi(x)|^2 \geq \left(\liminf_{|x|\to\infty}\Lambda_R(x,H_\ep) - \delta\right)
\|\phi\|^2_{\ell^2}\; .
\end{equation}
By Lemma \ref{perslem1} we know that for all $\delta>0$ and $\phi\in c_0\left(\disk\right)$ there exists $R_\delta$ such
that for all $R>R_\delta$
\begin{equation}\label{pers19}
\skpd{H_\ep\phi}{\phi} \geq \sum_{x\in\disk} \left(\Lambda_R(x,H_\ep) - \delta\right) |\phi(x)|^2 \; .
\end{equation}
Inserting \eqref{pers19} in \eqref{pers18} it follows that for all $\delta>0$ there exists $R_\delta$ such
that for all $R>R_\delta$ there exists $R_0$ such that for all $\phi\in c_0\left(\disk\setminus \overline{B_0(R_0)}\right)$
\begin{equation}\label{pers20}
\frac{\skpd{H_\ep\phi}{\phi}}{\|\phi\|_{\ell^2}^2} \geq \liminf_{|x|\to\infty} \Lambda_R(x,H_\ep) - 2 \delta \, .
\end{equation}
By the definition of $\Sigma (H_\ep)$ it follows directly that
\begin{equation}\label{pers21}
\Sigma (H_\ep) \geq \inf\left\{ \frac{\skpd{H_\ep \phi}{\phi}}{\|\phi\|_{\ell^2}^2}\, |\,
\phi\in c_0\left(\disk\setminus \overline{B_0(R_0)} \right)\right\}\, .
\end{equation}
The equation \eqref{pers19} holds for all $\phi \in c_0\left(\disk\setminus \overline{B_0(R_0)}\right)$, thus
we can take on the left hand side the infimum over all these functions, which together with \eqref{pers21} yields
\begin{equation}\label{pers22}
\Sigma (H_\ep) \geq \liminf_{|x|\to\infty} \Lambda_R(x,H_\ep) - 2\delta\; .
\end{equation}
The left hand side is independent of $R$ and since the relation holds for all $R>R_\delta$, it is possible to take the
limit $R\to\infty$, which yields for all $\delta>0$
\[ \Sigma (H_\ep) \geq \lim_{R\to +\infty} \liminf_{|x|\to\infty} \Lambda_R(x,H_\ep) - 2\delta\; .  \]
Thus in the limit $\delta$ the estimate \eqref{pers17} follows.
\end{proof}

\begin{proof}[Proof of Theorem \ref{perssontheo}]

We discuss the cases $\Sigma(H_\ep) = \infty$ and $\Sigma (H_\ep)< \infty$ separately.\\

{\sl Case 1:} $\Sigma (H_\ep)<\infty$:\\
As in the preceding proof, we conclude the equality by showing that both inequalities hold.\\

{\sl Step 1:} Estimate from below
\begin{equation}\label{pers23}
\inf \sigma_{ess} (H_\ep) \geq \Sigma (H_\ep)
\end{equation}

As a function of $R$, the term $\liminf_{|x|\to\infty} \Lambda_R(x,H_\ep)$ is monotonically decreasing, thus it follows
by Lemma \ref{perslem2} that for fixed $R>0$
\[ \Sigma (H_\ep) \leq  \liminf_{|x|\to\infty} \Lambda_R(x,H_\ep) \]
and thus for all $\delta>0$ there exists $a_\delta$ such that for all $x\in\disk$ with $|x|>a_\delta$
\begin{equation}\label{pers24}
\Sigma (H_\ep) - \frac{\delta}{2} \leq \Lambda_R(x,H_\ep)\; .
\end{equation}
On the other hand denoting by $\sigma (H_\ep)$ the spectrum of $H_\ep$, it is clear by the definition of
$\Lambda_R(x, H_\ep)$ and the Min-Max-principle that
\begin{equation}\label{pers25}
\Lambda_R(x,H_\ep) \geq \inf \sigma (H_\ep)\; .
\end{equation}
Since $H_\ep$ is bounded from below, it follows by \eqref{pers24} and \eqref{pers25} that there exists a constant $C>0$ such that for all $x\in\disk$
\begin{equation}\label{pers26}
\Lambda_R(x,H_\ep) \geq \Sigma (H_\ep) - C \; .
\end{equation}
We choose a function $W\in c_0\left(\disk\right)$ such that
$W(x) \geq C$ for $|x|<a_\delta$ and $W(x) \geq 0$ everywhere. Then for $H_\ep + W$ it follows by Lemma
\ref{perslem1}, \eqref{pers24} and \eqref{pers26} that for $\phi\in c_0\left(\disk\right)$
\begin{align*}
\skpd{(H_\ep + W)\phi}{\phi} &
\geq \sum_{x\in\disk} (W(x) - \Lambda_R(x,H_\ep) - \frac{\delta}{2}) |\phi(x)|^2 \\
&\geq \sum_{|x|\leq a_\delta} (\Sigma (H_\ep) - \frac{\delta}{2})|\phi(x)|^2 +
\sum_{|x|> a_\delta} (W(x) + \Sigma (H_\ep) - \delta)|\phi(x)|^2 \\
&\geq (\Sigma (H_\ep) - \delta) \sum_{x\in\disk}  |\phi(x)|^2\; .
\end{align*}
Thus
\begin{equation}\label{pers27}
\inf\sigma_{ess} (H_\ep + W) \geq \inf \sigma (H_\ep + W) \geq \Sigma (H_\ep) - \delta \; ,
\end{equation}
where the first estimate follows directly by the definition of the spectra.
The perturbation $W$ is compactly supported, thus each
$u\in \ell^2(\disk)$ is mapped by $W$ to a lattice function with compact support,
i.e. which is non-zero only at finitely many lattice points. Thus $W$ is a finite  rank operator and in particular
compact. Using Weyl's theorem (see e.g. Reed-Simon \cite{reed}), it follows that
\[ \sigma_{ess}(H_\ep + W) = \sigma_{ess}(H_\ep) \]
and since \eqref{pers27} holds for all $\delta>0$ the estimate \eqref{pers23} is shown.\\

{\sl Step 2:} Estimate from above
\begin{equation}\label{pers28}
\inf\sigma_{ess}(H_\ep) \leq \Sigma (H_\ep)
\end{equation}

Fix $\mu < \inf\sigma_{ess}(H_\ep)$ and denote by $\Pi_\mu:= \Pi_{(-\infty,\mu]}$ the spectral
projection to
the eigenspace of energies smaller or equal to $\mu$. Since $\mu$ lies below the essential spectrum
and $H_\ep$ is semi-bounded from below, it follows that $\Pi_\mu$ has finite rank. Thus there exists an
orthonormal system of eigenfunctions $\psi_1,\ldots,\psi_n\in \ell^2\left(\disk\right)$ such that
\[ \Pi_\mu = \sum_{j=1}^n \skpd{\,.\,}{\psi_j} \psi_j\]
and  for all $\delta>0$ there exists an $R_\delta$ such that
\[\sum_{|x|>R_\delta}|\psi_j(x)|^2  \leq \delta\; .\]
Therefore (by the Cauchy-Schwarz inequality)
for all $\phi \in c_0\left(\disk\setminus B_0(R_\delta)\right)$
\begin{equation}\label{pers29}
\|\Pi_\mu\phi(x)\|_{\ell^2}^2 = \sum_{j=1}^n |\skpd{\phi}{\psi_j}|^2 \leq \|\phi\|_{\ell^2}^2
\sum_{j=1}^n\sum_{|x|>R_\delta}|\psi_j(x)|^2  \leq \delta \|\phi\|_{\ell^2}^2\; .
\end{equation}
By the definition of $\Pi_\mu$ and since there exists a constant $C>0$ such that $H_\ep\geq -C$, we have
\begin{equation}\label{pers30}
\skpd{H_\ep\phi}{\phi} \geq \mu \skpd{(\id -\Pi_\mu)\phi}{(\id -\Pi_\mu)\phi} -
C \skpd{\Pi_\mu\phi}{\Pi_\mu \phi}\; .
\end{equation}
Therefore
\begin{eqnarray*}
\Sigma (H_\ep) &\geq& \inf\left\{ \frac{\skpd{H_\ep\phi}{\phi}}{\|\phi\|_{\ell^2}^2}\,|\, \phi
\in c_0\left(\disk\setminus B_0(R_\delta)\right)\right\} \\
&\geq& \inf \left\{ \mu \frac{\|(\id-\Pi_\mu)\phi\|^2}{\|\phi\|_{\ell^2}^2} - C\frac{\|\Pi_\mu\phi\|^2}{\|\phi\|^2}\,|\, \phi
\in c_0\left(\disk\setminus B_0(R_\delta)\right)\right\} \\
&=& \inf \left\{ \mu - (C+\mu)\frac{\|\Pi_\mu\phi\|^2}{\|\phi\|_{\ell^2}^2} \,|\, \phi
\in c_0\left(\disk\setminus B_0(R_\delta)\right)\right\}
\end{eqnarray*}
and by \eqref{pers29}
\[ \Sigma (H_\ep) \geq \mu - (C+\mu) \delta \; . \]
The left hand side is independent of $\delta$, thus for $\delta\to 0$ we get
\[ \Sigma (H_\ep) \geq \mu \]
for any $\mu<\inf\sigma_{ess} (H_\ep)$ and thus in the limit $\mu\to\inf\sigma_{ess} (H_\ep)$ the
estimate \eqref{pers28} follows and thus Theorem \ref{perssontheo} is proven.\\

{\sl Case 2:} $\Sigma (H_\ep) =\infty$:\\
By Lemma \ref{perslem2} it follows at once that $\lim_{|x|\to\infty} \Lambda_R(x,H_\ep) = \infty$, because
$\Lambda_R(x,H_\ep)$ is monotonically decreasing with respect to $R$. Thus for all $M>0$ there
exists a $a_M$ such that for all $x\in\disk$ with $|x|>a_M$ the estimate
$\Lambda_R(x,H_\ep) \geq M$ holds. On the other hand by \eqref{pers25} and the semi-boundedness of $H_\ep$ it follows that
there exists a constant $C>0$ such that
\[ \Lambda_R(x,H_\ep) \geq - C\, , \qquad \text{for all}\quad x\in\disk \; . \]
We can choose a function $W\in c_0\left(\disk\right)$ such that $W(x)\geq C+M$ for $|x|< a_M$ and $W(x)\geq 0$ everywhere.
Then
\[ \skpd{ (H_\ep + W)\phi}{\phi} \geq \skpd{ (W + \Lambda_R(.,H_\ep) - \frac{\delta}{2})\phi}{\phi} \geq
\left( M - \frac{\delta}{2}\right) \|\phi\|^2_{\ell^2} \]
and thus for all $M>0$ there exists a function $W\in c_0\left(\disk\right)$ such that
\[ \sigma_{ess} (H_\ep + W) \geq \sigma (H_\ep + W) \geq M\; . \]
As in the case $\Sigma (H_\ep)<\infty$ we have $\sigma (H_\ep + W) = \sigma (H_\ep)$ and therefore
$\sigma_{ess}(H_\ep)\geq M$ for all $M>0$ and thus $\sigma_{ess}(H_\ep) = \infty$.
\end{proof}

\end{appendix}

{\sl Acknowledgements.} The authors thank B. Helffer for many valuable discussions and remarks on the
subject of this paper. M.K. thanks F. Nier for bringing G\'erard-Nier \cite{gerard} to his attention.


\begin{thebibliography}{99}





\bibitem{agmon} S. Agmon: {\sl Lectures on Exponential Decay of
Solutions of Second-order Elliptic Equations: Bounds on
Eigenfunctions of N-Body Schr\"odinger Operators}, Mathematical
Notes 29, Princeton University Press, 1982


\bibitem{baake} E.Baake, M.Baake, A.Bovier, M.Klein: {\sl An asymptotic maximum principle
for essentially linear evolution models},  J. Math. Biol. 50 no.1, p. 83-114.  2005


\bibitem{begk1} A. Bovier, M. Eckhoff, V. Gayrard, M. Klein: {\sl Metastability in
stochastic dynamics of disordered mean-field models},
Probab. Theory Relat. Fields 119, p. 99-161, 2001

\bibitem{begk2} A. Bovier, M. Eckhoff, V. Gayrard, M. Klein: {\sl Metastability and
low lying spectra in reversible Markov chains}, Comm. Math. Phys. 228, p. 219-255, (2002)





\bibitem{calderon} A. P. Calderon, R. Vaillancourt: {\sl On the Boundedness of
Pseudo-Differential Operators}, J.Math.Soc. Japan 23,2,  p.
374-378, 1971

\bibitem{simon} H. L. Cycon, R. G. Froese, W. Kirsch, B. Simon: {\sl Schr\"odinger
Operators with Application to Quantum Mechanics
                and Global Geometry}, Springer, 1987

\bibitem{dima} M. Dimassi, J. Sj\"ostrand: {\sl Spectral Asymptotics in the Semi-
Classical Limit}, London Mathematical Society
                             Lecture Note Series 268,
                 Cambridge University Press, 1999




\bibitem{gerard} C. G\'erard, F. Nier: {\sl
Scattering theory for the perturbations of periodic Schr\"odinger
operators}  J. Math. Kyoto Univ.  38  (1998),  no. 4, 595--634.




\bibitem{grisjo} A. Grigis, J. Sj\"ostrand: {\sl Microlocal Analysis for Differential Operators},
London Mathematical Society, Lecture Note Series 196, Cambridge
University Press, 1994

\bibitem{helf} B. Helffer: {\sl Semi-Classical Analysis for the Schr\"odinger Operator
and Applications}, LNM 1336, Springer,  1988

\bibitem{helf2} B. Helffer: {\sl Spectral Theory and application}, Cours de
DEA 1999-2000

\bibitem{hesjo} B.Helffer, J.Sj\"ostrand: {\sl Multiple wells in the
               semi-classical limit I}, Comm. in P.D.E. 9 (1984), p. 337-408



\bibitem{hormander2} L. H\"ormander: {\sl The Analysis of Linear Partial Differential Operators 1}, Springer-Verlag
Berlin, 1983

\bibitem{hwang} I. L. Hwang: {\sl The $L^2$-Boundedness of Pseudodifferential Operators},
Trans.Amer.Math.Soc. 302, p. 55-76, 1987

\bibitem{roklein} M. Klein, E. Rosenberger: {\sl  Agmon-Type Estimates for a class of Difference Operators }, to appear
in Ann Inst. H. Poincare

\bibitem{martinez} A. Martinez: {\sl An Introduction to Semiclassical and Microlocal Analysis}, Springer-Verlag, 2002

\bibitem{reed} M. Reed, B. Simon: {\sl Methods of Modern Mathematical Physics 4},
Academic Press, 1979

\bibitem{robert} D. Robert: {\sl Autour de l'Approximation Semi-Classique}, Progr.
in Math.68. Birkh\"auser, 1987


\bibitem{thesis} E. Rosenberger: {\sl Asymptotic Spectral Analyis and Tunneling for a class
of Difference Operators}, Thesis, http://nbn-resolving.de/urn:nbn:de:kobv:517-opus-7393




\bibitem{Si1} B. Simon: {\sl Semiclassical analysis of low lying eigenvalues.I.
Nondegenerate minima: asymptotic expansions}, Ann Inst. H. Poincare Phys. Theor. 38,
p. 295 - 308, 1983





\end{thebibliography}
\end{document}